\documentclass[prd,
reprint,
onecolumn,
notitlepage,
longbibliography,
]{revtex4-1}

\usepackage{amssymb}
\usepackage{amsmath}
\usepackage{amsthm}
\usepackage{latexsym}
\usepackage{mathrsfs}
\usepackage{eucal}
\usepackage{tipa}
\usepackage{graphicx}
\usepackage{slashed}
\usepackage{amsbsy}

\usepackage{color}

\usepackage[pdftex,plainpages=false,pdfpagelabels,verbose=true]{hyperref}

\newcounter{mnotecount}[section]

\allowdisplaybreaks[2]

\theoremstyle{plain}
\newtheorem{theorem}{Theorem}[section]

\newtheorem{lemma}[theorem]{Lemma}
\newtheorem{corollary}[theorem]{Corollary}

\newtheorem{definition}[theorem]{Definition}
\newtheorem{remark}[theorem]{Remark}
\newtheorem{example}[theorem]{Example}

\def\sDiv{\mathscr{D}}
\def\sCurl{\mathscr{C}}
\def\sCurlDagger{\mathscr{C}^\dagger}
\def\sTwist{\mathscr{T}}

\DeclareMathOperator{\tho}{\text{\rm\textthorn}}
\DeclareMathOperator{\edt}{\eth}

\newcommand{\MidMet}{\mathcal{M}}
\renewcommand{\Re}{\mathrm{Re}}
\renewcommand{\Im}{\mathrm{Im}}
\def\ImA{\mathrm{Im}\mathcal{A}}

\newcommand{\Hcal}{\mathcal{H}}

\newcommand{\gfrak}{\mathfrak{g}}

\newcommand{\NPl}{l}
\newcommand{\NPn}{n}
\newcommand{\NPm}{m}
\newcommand{\NPmbar}{\bar m}

\newcommand{\Lie}{\mathcal L}
\newcommand{\half}{\frac{1}{2}}

\newcommand{\SL}{\mathrm{SL}}
\newcommand{\SO}{\mathrm{SO}}

\newcommand{\Co}{{\mathbb C}} 
\newcommand{\Spin}{\text{Spin}}	
\newcommand{\met}{g} 

\newcommand{\SymSpinSec}{\mathcal S}

\newcommand{\sfrak}{\mathfrak{s}}

\def\Mass{M}

\begin{document}

\title{New identities for linearized gravity on the Kerr spacetime}

\author{Steffen Aksteiner} \email{steffen.aksteiner@aei.mpg.de}
\affiliation{Albert Einstein Institute, Am M\"uhlenberg 1, D-14476 Potsdam,
Germany}

\author{Lars Andersson} \email{laan@aei.mpg.de}
\affiliation{Albert Einstein Institute, Am M\"uhlenberg 1, D-14476 Potsdam,
  Germany}

\author{Thomas B\"ackdahl} 
\email{thobac@chalmers.se}
\affiliation{Mathematical Sciences, Chalmers University of Technology and University of Gothenburg, SE-412 96 Gothenburg, Sweden}
\affiliation{School of Science and Technology, \"Orebro University, SE-701 82 \"Orebro, Sweden}
\affiliation{The University of Edinburgh, James Clerk Maxwell Building, Peter Guthrie Tait Road, Edinburgh,  EH9 3FD, UK}


\begin{abstract}
In this paper we derive a differential identity for linearized gravity on the Kerr spacetime and more generally on vacuum spacetimes of Petrov type D. We show that a linear combination of second derivatives of the linearized Weyl tensor can be formed into a complex symmetric 2-tensor $\MidMet_{ab}$ which solves the linearized Einstein equations. The identity makes this manifest by relating $\MidMet_{ab}$ to two terms solving the linearized Einstein equations by construction. The self-dual Weyl curvature of $\MidMet_{ab}$ gives a covariant version of the Teukolsky-Starobinsky identities for linearized gravity which, in addition to the two classical identities for linearized Weyl scalars with extreme spin weights, includes three additional equations. In particular, they are not consequences of the classical Teukolsky-Starobinsky identities, but are additional integrability conditions for linearized gravity. The result has direct application in the construction of symmetry operators and also yields a set of non-trivial gauge invariants for linearized gravity.
\end{abstract} 

\maketitle

\section{Introduction}

In this paper we prove a new identity for linearized gravity on the Kerr spacetime, and more generally for spacetimes of Petrov type D. Let $\dot g_{ab}$ be a solution to the source-free linearized Einstein equations and let $\dot{\Psi}_0, \dot{\Psi}_4$ be the corresponding gauge invariant extreme Weyl scalars. The Teukolsky master equations (TME) are two wave equations for $\dot{\Psi}_0, \dot{\Psi}_4$, respectively, while the classical Teukolsky-Starobinsky identities (TSI) are two fourth order differential relations between them. 

From up to second derivatives of $\dot{\Psi}_0, \dot{\Psi}_4$ we construct a complex symmetric 2-tensor $\MidMet_{ab}$ solving the linearized Einstein equations. The main result of this paper (here specialized to a Kerr background) is an identity of the form
\begin{equation}\label{eq:MainIdentity_TensorVersion} 
\MidMet_{ab} =  
\nabla_{(a}\mathcal{A}_{b)} + \frac{\Mass}{54} \mathcal{L}_{\xi}\dot g_{ab},
\end{equation} 
with $\mathcal{L}_{\xi}\dot g_{ab}$ being the Lie derivative of  $\dot g_{ab}$ along the time-like isometry $\xi$, $\Mass$ the Mass and the vector field $\mathcal{A}^a$ being defined in terms of up to three derivatives of the linearized metric $\dot g_{ab}$. The classical TSI can be derived covariantly as the extreme components of the self-dual Weyl curvature of \eqref{eq:MainIdentity_TensorVersion}. In fact, as we shall see, one finds not only the two classically known Teukolsky-Starobinsky identities, but three additional identities, which are presented here for the first time. 

The \citet{kerr:1963PhRvL..11..237K} family of rotating, stationary, vacuum spacetimes, with parameters mass $M$ and angular momentum per unit mass $a$ can be characterized as the asymptotically flat vacuum Petrov type D spacetimes with positive mass, cf. \cite{2015arXiv150402069A}, see also \cite{stephani:etal:2009esef.book.....S, *Penrose:1986fk} for background. If $|a| \leq M$, the maximally extended Kerr spacetime contains a black hole, which in the subextreme case $|a| < M$ is non-degenerate, with a bifurcate event horizon. The Kerr rotating black hole model, which plays a central role in astrophysics, and gravitational wave research, is conjectured to be unique and stable. The black hole uniqueness conjecture states that an asymptotically flat, stationary, vacuum spacetime containing a non-degenerate horizon belongs to the subextreme Kerr family. The black hole stability conjecture states that the subextreme Kerr family is dynamically stable, i.e. that a small perturbation of Kerr Cauchy data generates a Cauchy development which is asymptotic to the future to a member of the Kerr family. Both of these conjectures remain open, in spite of intense work during the last decades. The black hole uniqueness and stability problems, and related problems including scattering are among the most important problems in general relativity and the special features of the Kerr geometry, which are to a large extent due to its algebraically special nature play a fundamental role in most of the work on these problems.

As was shown by \citet{walker:penrose:1970CMaPh..18..265W} any vacuum spacetime of Petrov type D admits a Killing spinor $\kappa_{AB}$ of valence 2, or equivalently a conformal Killing-Yano tensor. 
The Kerr geometry belongs to the generalized Kerr-NUT subclass \citep{ferrando:saez:2007JMP....48j2504F}, characterized by the condition that the Killing vector field $\nabla_{A'}{}^B \kappa_{AB}$ is proportional to a real Killing vector field, or equivalently, by the condition that it admits a Killing-Yano tensor. If $Y_{ab}$ is a Killing-Yano tensor, then $K_{ab} = Y_{ac} Y^c{}_b$ is a Killing tensor. In the case of Kerr, $K_{ab}$ is proportional to the Carter Killing tensor \citep{carter:1968PhRv..174.1559C},  and gives rise to a conserved quantity for geodesics, the Carter constant. In addition to the just mentioned facts, field equations on Petrov type D have several further useful and nontrivial properties including decoupling, existence of non-trivial symmetries, and separation of variables, as well as associated conservation laws.  

The TME \citep{teukolsky:1973} and the TSI \citep{1974ApJ...193..443T, 1974JETP...38....1S} which play a crucial role in the analysis of higher spin fields on the Kerr background, generalize to the Petrov type D class. For a spin-$\sfrak$ field, $\sfrak = 1, 2$, where the spin-$1$ field is a Maxwell test field, and the spin-$2$ field linearized gravity, the TME are wave equations governing the Newman-Penrose scalars \citep{1962JMP.....3..566N, *GHP}, defined with respect to a principal null tetrad, with extreme spin weights $\pm \sfrak$. Decoupling refers to the fact, valid in the whole Petrov type D class, that for each $\sfrak$, the spin-$\sfrak$ field equation implies a pair of decoupled scalar TME. For spacetimes in the generalized Kerr-NUT subclass, the TME admit separation of variables, which is a consequence of the existence of non-trivial commuting symmetry operators.  The existence of commuting symmetry operators for the TME, and more generally symmetry operators taking solutions of the spin-$\sfrak$ field equation to solutions, are also consequences of the Petrov type D or generalized Kerr-NUT condition, cf. \citep{ABB:symop:2014CQGra..31m5015A,2016arXiv160904584A}. 

In its classical form on Kerr, the TSI  relate the solutions of the radial Teukolsky equations for fields of spin-weights $\pm \sfrak$ on the Kerr spacetime, and are thus valid only for the separated form of the TME. In that context the TSI are sometimes referred to as the Teukolsky-Press, or Starobinsky-Churilov identities. For the case of linearized gravity on a Petrov type D vacuum spacetime, a derivation of the TSI using the Newman-Penrose formalism, which does not require a separation of variables, was given by \citet{TorresDelCastillo:1994}, later corrected by \citet{SilvaOrtigoza:1997}. See the paper by \citet{Whiting:2005hr} for discussion and background. 

The TME and TSI are consequences of the spin-$\sfrak$ field equations and may thus be viewed as integrability conditions. As pointed out by 
\citet{coll:etal:1987JMP....28.1075C}, in the spin-1 case the classical TSI system must be completed by adding one equation in order for the system of integrability conditions given by the TME and TSI systems to be \emph{equivalent} to the Maxwell system, modulo charge. Examining the full TSI system, one finds that those equations which correspond to the classical TSI have extreme spin weights. For this reason we shall use the term \emph{extreme TSI} when referring to the classical form of the TSI. Due to the fact that the spin-1 TME and TSI involve only the Maxwell scalars of extreme spin-weights,  the non-radiating mode carrying the charge, also known as the Coulomb solution, cancels out of the TME and TSI systems. Thus, in order to reconstruct a Maxwell field from a solution of the TME and full TSI systems, it is necessary to specify the charge as an additional parameter.

In order to understand how to derive the full TSI for the spin-2 case, the following remarks are helpful. In the Maxwell case, the Debye potential construction \citep{wald:1978PhRvL..41..203W} on the Kerr background can be used to construct from the Maxwell field a complex pure gauge vector potential 
$$
\boldsymbol{\alpha}_a = (df)_a
$$
given by a first order differential operator acting on the Maxwell field strength. 
The field strength of $\boldsymbol{\alpha}_a$ has vanishing complex anti-self dual and self dual parts, which correspond to the TME and full TSI systems respectively. 
In the case of linearized gravity, the role of the vector potential is played by the 
linearized metric. Here, the analogous situation holds. The Debye potential construction can be used to construct from a solution $\dot g_{ab}$ of the linearized vacuum Einstein equations\footnote{We shall sometimes refer to the linearized vacuum Einstein equations as the source-free linearized Einstein equations.}, a complex, traceless, symmetric 2-tensor $\MidMet_{ab}$ which is essentially a pure gauge metric satisfying the linearized vacuum Einstein equations. This result is the main theorem of the present paper, cf. Theorem \ref{thm:MainThmTensorVersionIntro}. 

The self dual Weyl curvature of $\MidMet_{ab}$ yields the full TSI for linearized gravity on vacuum spacetimes of Petrov type D, and in particular on the Kerr spacetime. The full TSI system is a differential relation of order four in the linearized Weyl curvature. We also note that the anti-self dual Weyl curvature of $\MidMet_{ab}$ yields a fourth order identity related to the TME. A new feature is encountered compared to the spin-1 case, since the resulting identities contain terms involving the Lie derivative of the background curvature. 

As just mentioned, the intermediate metric $\MidMet_{ab}$ is defined in terms of Debye potentials for the linearized vacuum Einstein equations. We now recall this construction, first restricting to Minkowski space. Following \citet{1958PhRv..112..674S}, let $\Hcal_{abcd}$ be an anti-self dual Weyl field\footnote{Here we use a complex anti-self dual Weyl field for consistence with the rest of the paper, although this is not used in \cite{1958PhRv..112..674S}.}, i.e. a tensor with the symmetries of the Riemann tensor, 
$\Hcal_{abcd}  = \Hcal_{[ab]cd} = \Hcal_{cdab}$, $\Hcal_{[abc]d} = 0$, satisfying $\Hcal^a{}_{bac} = 0$ and $\half \epsilon_{ab}{}^{ef} \Hcal_{efcd} = -i \Hcal_{abcd}$, 
and let 
\begin{equation}\label{eq:BS}
\gfrak_{ab} = \nabla^c \nabla^d \Hcal_{acbd} .
\end{equation} 
Then, if $\nabla^e \nabla_e \Hcal_{abcd} = 0$, it follows  that $\gfrak_{ab}$ solves the linearized vacuum Einstein equation.

The analogous construction for massless spin-$\sfrak$ fields on the 4-dimensional Minkowski space was discussed by  \citet{penrose:1965}. In \cite{2014CMaPh.331..755A} this was used to prove decay estimates for such fields, based on decay estimates for the wave equation. We shall now describe the analogue of the Sachs-Bergmann construction in the case of a vacuum Petrov type D metric. 

Introduce the following complex anti-self dual tensors with the symmetries of the Weyl tensor  
\begin{align*} 
Z^0_{abcd} ={}& 4 \NPmbar_{[a}\NPn_{b]} \NPmbar_{[c} \NPn_{d]}, \\ 
Z^4_{abcd} ={}& 4 \NPl_{[a}\NPm_{b]} \NPl_{[c} \NPm_{d]} ,
\end{align*} 
where $(\NPl^a, \NPn^a, \NPm^a, \NPmbar^a)$ constitutes a principal null tetrad. 
These are analogues of the anti-self dual bivectors $Z^0_{ab}, Z^2_{ab}$, see \cite[\S 2]{aksteiner:andersson:2013CQGra..30o5016A}. 
For a complex scalar $\chi_0$, let $\Hcal_{abcd}$ be given by 
\begin{equation}\label{eq:Hcal0} 
\Hcal_{abcd} = \kappa_1^4 \chi_0  Z^0_{abcd} .
\end{equation}
with the complex function $\kappa_1$ being the Killing spinor coefficient, see \eqref{eq:TypeDKS} for details. Define the 1-form $U_a$ by
$$
U_a = - \nabla_a \log (\kappa_1) ,
$$ 
cf. \eqref{eq:UU11Def},
and consider the following analogue\footnote{More precisely it differs by a gauge transformation of \textit{third kind}, cf. \cite{1979PhRvD..19.1641K}, so that the scalar potential solves the TME. } of \eqref{eq:BS}, which sends $\Hcal_{abcd}$ to a 2-tensor $\gfrak_{ab}$, 
$$
\gfrak_{ab} = 
\nabla^{c}(\nabla_{d}+4 U_{d}) \Hcal_{(a}{}^{d}{}_{b)c} .
$$
A calculation shows that $\gfrak_{ab}$ is a complex solution to the linearized vacuum Einstein equation provided the scalar $\kappa_1^4 \chi_0$ solves the TME for spin weight $+2$, cf.  \cite{1979PhRvD..19.1641K}. See corollary~\ref{cor:covTME} below for the covariant form of the TME system,  see also equation \eqref{eq:TME:spin-2} for the component form. The analogous construction with  
\begin{equation}\label{eq:Hcal4}
\Hcal_{abcd} = \kappa_1^4 \chi_4 Z^4_{abcd}
\end{equation}
yields a solution to the linearized Einstein equation in the same way, provided that now the scalar $\kappa_1^4 \chi_4$ solves the TME for spin weight $-2$. Note that in general, the linearized metrics $\gfrak_{ab}$ constructed from \eqref{eq:Hcal0} and \eqref{eq:Hcal4} are different. 
We are now able to state the tensor version of our main result, which describes this difference. Given scalars $\dot \Psi_0, \dot \Psi_4$ of spin weights $2$ and $-2$ respectively, define 
\begin{equation}\label{eq:Hcalpm-def}
\Hcal^\pm_{abcd} = \kappa_1^4 \dot \Psi_0 Z^0_{abcd} \pm \kappa_1^4 \dot \Psi_4 Z^4_{abcd} . 
\end{equation} 
\begin{theorem}[Tensor version]\label{thm:MainThmTensorVersionIntro} Let $\dot g_{ab}$ be a solution to the source-free linearized Einstein equation on a vacuum background of Petrov type D, and let $\dot \Psi_0, \dot \Psi_4$ be the linearized Weyl scalars of spin weights $\pm 2$ defined with respect to a principal null tetrad. Let 
\begin{equation}\label{eq:Mabdef:tensor}
\mathcal{M}_{ab} = \nabla^{c}(\nabla_{d}+4 U_{d}) \Hcal^{-}\negmedspace{}_{(a}{}^{d}{}_{b)c} .
\end{equation}
Then, there is a complex vector field $\mathcal{A}_a$ depending on up to three derivatives of the linearized metric $\dot g_{ab}$, such that  
\begin{equation}\label{eq:nice-tensor} 
\mathcal{M}_{ab} =  
\nabla_{(a}\mathcal{A}_{b)} + \tfrac{1}{2} \Psi_{2} \kappa_1^3 \mathcal{L}_{\xi}\dot g_{ab}.
\end{equation} 
Here, $\xi^a$ is a Killing vector defined in \eqref{eq:XiDef} and $\Psi_2$ is the only non-vanishing component of the background curvature.
\end{theorem} 
\begin{remark} 
In the Maxwell case, the analogue of \eqref{eq:nice-tensor} is that the vector potential arising out of the Debye potential construction by taking the difference of the extreme Maxwell scalars from \emph{the same} Maxwell field, is pure gauge, cf. \citet[eq. (5.40)]{aksteiner:thesis}. In the spin-2 case above, the term involving $\Lie_\xi \dot g_{ab}$ is a new feature, which indicates an important qualitative difference between the spin-1 and spin-2 cases. 
\end{remark} 

In a Petrov type D vacuum spacetime we have that $\xi^a$ is Killing and $\Psi_2 \kappa_1^3$ is constant, see \eqref{eq:TypeDKS} below. Hence, the intermediate metric $\mathcal{M}_{ab}$ given in \eqref{eq:nice-tensor} is a complex solution of the linearized Einstein equation\footnote{The first term on the right-hand side of equation \eqref{eq:nice-tensor} is pure gauge since it is the action of a linearized diffeomorphism generated by $\mathcal{A}_a$.}.

Let us finally mention some further properties and immediate applications of the result. From \eqref{eq:MainIdentity_TensorVersion} it follows that $\nabla_{(a}Im\mathcal{A}_{b)}$ is gauge invariant (since the LHS is gauge invariant and the second term on the RHS is real). In fact, as we show in Lemma~\ref{lem:GaugeDependence}, $Im\mathcal{A}_{b}$ itself is gauge invariant and contributes to a complete set of invariants, as explored by two of the authors in \cite{2018PhRvL.121e1104A}. It should also be noted, that \eqref{eq:MainIdentity_TensorVersion} can be viewed as an operator identity applied to $\dot{g}_{ab}$. This perspective, together with properties of its formal adjoint identity has been made use of by two of the authors in \cite{2016arXiv160904584A} to derive a symmetry operator for the linearized Einstein equations on type D spacetimes.

\subsection*{Overview of this paper} In section~\ref{sec:prel} we give some background and preliminary results. Section~\ref{sec:2-spinors} introduces the 2-spinor formalism which shall be used throughout the paper, and section~\ref{sec:TypeDstructure} contains a review of the consequences of the existence of a Killing spinor on vacuum type D spacetimes. In sections~\ref{sec:ExtFundSpinOp} and \ref{sec:projop} we introduce a set of geometrically defined operators together with commutation rules, which allow us to exploit the special geometry in Petrov type D spacetimes. In section~\ref{sec:LinearizedGravity} a spinorial form of the field equations of linearized gravity is presented. We derive a convenient form of the linearized Bianchi identity in corollary~\ref{cor:cD4-id}. This equation plays a central role in the proof of the main theorem given in section~\ref{sec:MainTheorem}. 
In lemma~\ref{lem:Mcurvature} we analyze the curvature of the intermediate metric $\MidMet_{ab}$. In particular, its self dual linearized Weyl curvature gives a covariant form of the full spin-2 TSI.  
Finally, in corollary~\ref{cor:tomain:intro} we give a simplified form of the TSI for the Kerr case, containing only gauge invariant quantities. 
Appendix~\ref{secApp:GHPForm} contains the GHP component form of various equations and appendix~\ref{app:KopComm} states commutator relations of certain spinorial operators  discussed in this paper.

\section{Preliminaries} \label{sec:prel} 
\subsection{2-spinors and irreducible decompositions} \label{sec:2-spinors}

Let $(\mathcal{N}, \met_{ab})$ be a 3+1 dimensional Lorentzian spin manifold with real metric $\met_{ab}$ of signature $+---$. The spacetimes we are interested in here are spin, in particular any orientable, globally hyperbolic 3+1 dimensional spacetime is spin, cf. \citet[page 346]{Ger70spinstructII}. We shall throughout the paper make use of 2-spinor formalism, which simplifies calculations and makes many geometrical structures more transparent. See  \cite{Penrose:1986fk} and \cite{2015arXiv150402069A} for background.

If $\mathcal{N}$ is spin, then the orthonormal frame bundle $\SO(\mathcal{N})$ admits a lift to $\Spin(\mathcal{N})$, a principal $\SL(2,\Co)$-bundle. The group $\SL(2,\Co)$ has two inequivalent fundamental representations $\Co^2$ and $\bar \Co^2$. We denote sections of the corresponding spinor bundles with unprimed and primed uppercase indices, respectively. The action of $\SL(2,\Co)$ leaves invariant an anti-symmetric 2-spinor $\epsilon_{AB}$, the spin-metric.

The associated bundle construction now gives vector bundles over $\mathcal{N}$ corresponding to the  representations of $\SL(2,\Co)$, in particular we have bundles of valence $(k,l)$ spinors with sections $\varphi_{A\cdots K A' \cdots L'}$. Here $k$ and $l$ are the numbers of unprimed and primed indices. An important aspect of the 2-spinor formalism is the correspondence between tensors and spinors. An example is provided by the correspondence between metric and spin-metric $g_{ab} = \epsilon_{AB} \bar \epsilon_{A'B'}$. The Levi-Civita connection lifts to act on sections of the spinor bundles, 
\begin{equation}\label{eq:nablavarphi}
\nabla_{AA'} : \varphi_{B \cdots D B' \cdots D' } \to \nabla_{AA'} \varphi_{B \cdots D B' \cdots D'},
\end{equation} 
where we have used the tensor-spinor correspondence to replace the index $a$ by $AA'$. 

Irreducible representations of $\SL(2,\Co)$ and hence also of $\SO(1,3)$ correspond exactly to symmetric spinors, which are automatically traceless.  The space of symmetric spinors of valence $(k,l)$ is denoted by $\SymSpinSec_{k,l}$. The correspondence between symmetric spinors and irreducible representations of $\SL(2,\Co)$ yields efficient methods for decomposition of geometric expressions into irreducible pieces, which can be used for canonicalization. The  \emph{SymManipulator} package \citep{Bae11a}, which has been developed by one of the authors (T.B.) for the \emph{Mathematica} based symbolic differential geometry suite \emph{xAct} \citep{xAct}, exploits in a systematic way the above mentioned decompositions and allows one to carry out investigations which are not feasible to do by hand.  The related 
\emph{SpinFrames} package \citep{SpinFrames} developed by two of the authors implements computations in tetrad components using the Newman-Penrose (NP) and Geroch-Held-Penrose (GHP) \citep{GHP} formalisms.

The above mentioned correspondence between spinors and tensors, and the decomposition into irreducible pieces, can be applied to the Riemann curvature tensor. In this case, they correspond to the scalar curvature $R$, traceless Ricci tensor $S_{ab}$, and the Weyl tensor $C_{abcd}$. The Riemann tensor then takes the form  
\begin{align}
R_{abcd}={}&- \tfrac{1}{12} g_{ad} g_{bc} R
 + \tfrac{1}{12} g_{ac} g_{bd} R
 + \tfrac{1}{2} g_{bd} S_{ac}
 -  \tfrac{1}{2} g_{bc} S_{ad}
 -  \tfrac{1}{2} g_{ad} S_{bc}
 + \tfrac{1}{2} g_{ac} S_{bd}
 + C_{abcd},
\end{align}
and the spinor equivalents of these tensors are
\begin{subequations}
\begin{align}
C_{abcd}={}&\Psi_{ABCD} \bar\epsilon_{A'B'} \bar\epsilon_{C'D'}+\bar\Psi_{A'B'C'D'} \epsilon_{AB} \epsilon_{CD},\\
S_{ab} ={}& -2 \Phi_{ABA'B'},\\
R={}&24 \Lambda.
\end{align}
\end{subequations}
The irreducible decomposition into symmetric spinors in particular applies to covariant derivatives of symmetric spinors $\varphi_{A\cdots D A' \cdots D'} \in \SymSpinSec_{k,l}$. Decomposing \eqref{eq:nablavarphi} into its irreducible parts leads to
\begin{align}
\nabla_A{}^{A'} \varphi_{A_1 \cdots A_k}{}^{A'_1 \cdots A'_l} ={}& 
(\sTwist \varphi)_{AA_1 \cdots A_k}{}^{A'A'_1 \cdots A'_l} - \tfrac{l}{l+1} \bar{\epsilon}^{A'(A'_1}(\sCurl \varphi)_{AA_1 \cdots A_k}{}^{A'_2 \cdots A'_l)} \nonumber \\
&- \tfrac{k}{k+1}\epsilon_{A(A_1}(\sCurlDagger \varphi)_{A_2 \cdots A_k)}{}^{A'A'_1 \cdots A'_l} + \tfrac{kl}{(k+1)(l+1)}\epsilon_{A(A_1}\bar{\epsilon}^{A'(A'_1}(\sDiv \varphi)_{A_2 \cdots A_k)}{}^{A'_2 \cdots A'_l)}
\end{align}
with coefficients given by the following four \emph{fundamental spinor operators} \cite[\S 2.1]{ABB:symop:2014CQGra..31m5015A}, also implemented in the \emph{SymManipulator} package  \citep{Bae11a}.
\begin{definition}
The differential operators
$$
\sDiv:\mathcal{S}_{k,l}\rightarrow \mathcal{S}_{k-1,l-1}, \quad 
\sCurl:\mathcal{S}_{k,l}\rightarrow \mathcal{S}_{k+1,l-1}, \quad 
\sCurlDagger:\mathcal{S}_{k,l}\rightarrow \mathcal{S}_{k-1,l+1}, \quad 
\sTwist:\mathcal{S}_{k,l}\rightarrow \mathcal{S}_{k+1,l+1}
$$
are defined as
\begin{subequations} \label{eq:FundamentalOperators}
\begin{align}
(\sDiv \varphi)_{A_1\dots A_{k-1}}{}^{A_1'\dots A_{l-1}'}\equiv{}&
\nabla^{BB'}\varphi_{A_1\dots A_{k-1}B}{}^{A_1'\dots A_{l-1}'}{}_{B'},\\
(\sCurl \varphi)_{A_1\dots A_{k+1}}{}^{A_1'\dots A_{l-1}'}\equiv{}&
\nabla_{(A_1}{}^{B'}\varphi_{A_2\dots A_{k+1})}{}^{A_1'\dots A_{l-1}'}{}_{B'},\\
(\sCurlDagger \varphi)_{A_1\dots A_{k-1}}{}^{A_1'\dots A_{l+1}'}\equiv{}&
\nabla^{B(A_1'}\varphi_{A_1\dots A_{k-1}B}{}^{A_2'\dots A_{l+1}')},\\
(\sTwist \varphi)_{A_1\dots A_{k+1}}{}^{A_1'\dots A_{l+1}'}\equiv{}&
\nabla_{(A_1}{}^{(A_1'}\varphi_{A_2\dots A_{k+1})}{}^{A_2'\dots A_{l+1}')}.
\end{align}
\end{subequations}
The operators are called respectively the divergence, curl, curl-dagger, and twistor operators. 
\end{definition}
Note that in contrast to \cite{ABB:symop:2014CQGra..31m5015A} we suppress valence indices on the operators. With respect to complex conjugation, the operators satisfy $\overline{\sDiv} = \sDiv$, $\overline{\sTwist} = \sTwist$, $\overline{\sCurl} = \sCurlDagger$, $\overline{\sCurlDagger} = \sCurl$, but note that $\overline{\SymSpinSec_{k,l}}=\SymSpinSec_{l,k}$. Commutation formulas for the fundamental operators are given in \cite[\S 2.2]{ABB:symop:2014CQGra..31m5015A}. The Bianchi identity in terms of these operator is given by
\begin{subequations}
\begin{align}
(\sCurlDagger \Psi)_{ABCA'} ={}& (\sCurl \Phi)_{ABCA'}, \label{eq:Bianchi1} \\
(\sDiv \Phi)_{AA'} ={}& -3(\sTwist \Lambda)_{AA'}. \label{eq:Bianchi2}
\end{align}
\end{subequations}

\subsection{Geometric structure of Petrov type D spacetimes} \label{sec:TypeDstructure}
It is well known \citep{walker:penrose:1970CMaPh..18..265W} that vacuum spacetimes of Petrov type D admit a non-trivial irreducible symmetric 2-spinor $\kappa_{AB}$ solving the Killing spinor equation
\begin{align} \label{eq:Twistkappa2}
(\sTwist \kappa)_{ABCA'} = 0.
\end{align}
Defining the spinors
\begin{align}
\xi_{AA'}={}&(\sCurlDagger \kappa)_{AA'}, \label{eq:XiDef} \\
\lambda_{A'B'}={}&(\sCurlDagger \sCurlDagger \kappa)_{A'B'},
\end{align}
the complete table of derivatives reads
\begin{subequations}\label{eq:kappasys}
\begin{align}
\nabla_{AA'}\kappa_{BC}={}&- \tfrac{1}{3} \xi_{CA'} \epsilon_{AB}
 -  \tfrac{1}{3} \xi_{BA'} \epsilon_{AC},\\
\nabla_{AA'}\xi_{LL'}={}&- \tfrac{1}{2} \lambda_{A'L'} \epsilon_{AL}
 -  \tfrac{3}{4} \kappa^{BC} \Psi_{ALBC} \bar\epsilon_{A'L'}, \label{eq:kappasysb}\\
\nabla_{CC'}\lambda_{A'B'}={}&2 \xi_{C}{}^{D'} \bar\Psi_{A'B'C'D'}.
\end{align}
\end{subequations} 
The fact that the system of equations \eqref{eq:kappasys} for  $(\kappa_{AB}, \xi_{AA'}, \lambda_{A'B'})$ is closed (in the sense that the right hand side involves, in addition to background fields, only the fields $(\kappa_{AB}, \xi_{AA'}, \lambda_{A'B'})$ themselves and no derivatives), implies in particular that higher derivatives of $\kappa_{AB}$ do not produce any relations not derivable from \eqref{eq:kappasys}. Tensor symmetrizing \eqref{eq:kappasysb} leads to zero on the right-hand side and shows that $\xi^{AA'}$ is a Killing vector.
\begin{remark} 
If we furthermore assume the generalized Kerr-NUT condition that $\xi_{AA'}$ as defined in \eqref{eq:XiDef} is real, the middle equation \eqref{eq:kappasysb} simplifies due to
\begin{align}
\lambda_{A'B'}
={}&(\sCurlDagger \bar{\xi})_{A'B'} 
={}(\sCurlDagger \sCurl \bar{\kappa})_{A'B'} \nonumber\\
={}&-6 \Lambda \bar{\kappa}_{A'B'}
 + \tfrac{3}{2} \bar\Psi_{A'B'C'D'} \bar{\kappa}^{C'D'}
 -  \tfrac{3}{4} (\sDiv \sTwist \bar{\kappa})_{A'B'}\nonumber\\
 ={}&\tfrac{3}{2} \bar\Psi_{A'B'C'D'} \bar{\kappa}^{C'D'}
\end{align}
The complete table of derivatives \eqref{eq:kappasys} reduces to
\begin{subequations}
\begin{align}
\nabla_{AA'}\kappa_{BC}={}&- \tfrac{1}{3} \xi_{CA'} \epsilon_{AB}
 -  \tfrac{1}{3} \xi_{BA'} \epsilon_{AC},\\
\nabla_{AA'}\xi_{LL'}={}&- \tfrac{3}{4} \bar{\kappa}^{B'C'} \bar\Psi_{A'L'B'C'} \epsilon_{AL}
 -  \tfrac{3}{4} \kappa^{BC} \Psi_{ALBC} \bar\epsilon_{A'L'}.
\end{align}
\end{subequations}
\end{remark} 
Using a principal dyad \footnote{Note that $\kappa_1$ and $\Psi_2$ can be expressed covariantly via the relations $\kappa_{AB} \kappa^{AB} = -2 \kappa_{1}{}^2$ and $\Psi_{ABCD} \Psi^{ABCD}=6 \Psi_{2}^2$. Hence, we can allow $\kappa_1$ and $\Psi_2$ in covariant expressions.} the Weyl and Killing spinors take the form
\begin{subequations} 
\begin{align}
\Psi_{ABCD}={}&6 \Psi_{2} o_{(A}o_{B}\iota_{C}\iota_{D)}, \\
\kappa_{AB}={}&-2 \kappa_1 o_{(A}\iota_{B)}, \label{eq:TypeDKS}
\end{align}
\end{subequations} 
with 
\begin{align} \label{eq:kappa1ProptoPsi2}
\kappa_1 \propto \Psi_2^{-1/3} .
\end{align}  
In addition to the Killing vector field \eqref{eq:XiDef} another important vector field is defined by
\begin{align} \label{eq:UU11Def}
U_{AA'}={}&- \frac{\kappa_{AB} \xi^{B}{}_{A'}}{3 \kappa_1^2} = - \nabla_{AA'}\log(\kappa_1) .
\end{align}
We have the complete table of derivatives
\begin{subequations} 
\begin{align}
(\sDiv U)={}&-2 \Psi_{2}
 + \frac{\xi_{AA'} \xi^{AA'}}{9 \kappa_1^2}, \\
(\sCurl U)_{AB}={}&0, \\
 (\sCurlDagger U)_{A'B'}={}&0, \\
 (\sTwist U)_{AB}{}^{A'B'}={}&\frac{\kappa_{AB} (\sCurlDagger \xi)^{A'B'}}{6 \kappa_1^2}
 + 2 U_{(A}{}^{(A'}U_{B)}{}^{B')}
 -  \frac{\xi_{(A}{}^{(A'}\xi_{B)}{}^{B')}}{9 \kappa_1^2},
\end{align} 
\end{subequations}
because it is completely determined by the Killing spinor \eqref{eq:TypeDKS}. In particular $U_{AA'}$ is closed, $(dU)_{ab} = 0$. 

The curvature can be expressed in terms of the Killing spinor according to
\begin{align} \label{eq:PsiCDeTokappa}
\Psi_{ABCD}={}&\frac{3 \Psi_{2} \kappa_{(AB}\kappa_{CD)}}{2 \kappa_1^2}.
\end{align}
Applying $\sTwist$ and using \eqref{eq:Twistkappa2}, \eqref{eq:UU11Def}, \eqref{eq:kappa1ProptoPsi2} yields the relation
\begin{align}\label{eq:TwistPsiCDe} 
(\sTwist \Psi)_{ABCDFA'}={}&5 \Psi_{(ABCD}U_{F)A'}.
\end{align}

\begin{remark}\label{rem:explicit} 
On a Kerr spacetime with parameters $(M,a)$ in a principal tetrad in Boyer-Lindquist coordinates $(t,r,\theta,\phi)$, the curvature scalar is given by $\Psi_2=-M(r - i a \cos\theta)^{-3}$ and we can set $ \kappa_1 = -\tfrac{1}{3}(r - i a \cos\theta)$. Then one finds $\Psi_{2} \kappa_1^3={}\tfrac{1}{27} M$ and $\xi^a =  (\partial_t)^a$.
\end{remark}

\subsection{Extended fundamental spinor operators}\label{sec:ExtFundSpinOp}
As we often need to rescale with powers of $\kappa_1$ and $\bar\kappa_1$ we introduce extended fundamental operators with indices $n,m$:
\begin{subequations}\label{eq:ExtendedFundamentalOperators1} 
\begin{align}
 (\sDiv_{(n,m)} \varphi)_{A_1 \dots A_{k-1}}{}^{A'_1 \dots A'_{l-1}} &\equiv  \kappa_{1}^{n}\bar\kappa_{1}^{m} (\sDiv \kappa_{1}^{-n}\bar\kappa_{1}^{-m} \varphi)_{A_1 \dots A_{k-1}}{}^{A'_1 \dots A'_{l-1}}, \\
 (\sCurl_{(n,m)} \varphi)_{A_1 \dots A_{k+1}}{}^{A'_1 \dots A'_{l-1}} &\equiv  \kappa_{1}^{n}\bar\kappa_{1}^{m} (\sCurl \kappa_{1}^{-n}\bar\kappa_{1}^{-m} \varphi)_{A_1 \dots A_{k+1}}{}^{A'_1 \dots A'_{l-1}}, \\
 (\sCurlDagger_{(n,m)} \varphi)_{A_1 \dots A_{k-1}}{}^{A'_1 \dots A'_{l+1}} &\equiv  \kappa_{1}^{n}\bar\kappa_{1}^{m} (\sCurlDagger \kappa_{1}^{-n}\bar\kappa_{1}^{-m} \varphi)_{A_1 \dots A_{k-1}}{}^{A'_1 \dots A'_{l+1}}, \\
 (\sTwist_{(n,m)} \varphi)_{A_1 \dots A_{k+1}}{}^{A'_1 \dots A'_{l+1}} &\equiv  \kappa_{1}^{n}\bar\kappa_{1}^{m}(\sTwist \kappa_{1}^{-n}\bar\kappa_{1}^{-m}\varphi)_{A_1 \dots A_{k+1}}{}^{A'_1 \dots A'_{l+1}}.
\end{align}
\end{subequations}
For $n = m = 0$ it coincides with the definition \eqref{eq:FundamentalOperators} of the fundamental operators and the indices will be suppressed in that case. Because $U_{AA'}$ is a logarithmic derivative, see \eqref{eq:UU11Def}, the operators \eqref{eq:ExtendedFundamentalOperators1} can equivalently be expressed as 
\begin{subequations}\label{eq:ExtendedFundamentalOperators2}
\begin{align}
(\sDiv_{(n,m)} \varphi)_{A_1 \dots A_{k-1}}{}^{A'_1 \dots A'_{l-1}}
&=
\left[\nabla^{BB'} + n U^{BB'} + m \bar U^{BB'} \right] \varphi_{A_1 \dots A_{k-1} B}{}^{A'_1 \dots A'_{l-1}}{}_{B'} ,  
\\
(\sCurl_{(n,m)} \varphi)_{A_1 \dots A_{k+1}}{}^{A'_1 \dots A'_{l-1}}
&= 
\left[\nabla_{(A_1}{}^{B'} + n U_{(A_1}{}^{B'} + m \bar U_{(A_1}{}^{B'} \right] \varphi_{A_2 \dots A_{k+1})}{}^{A'_1 \dots A'_{l-1}}{}_{B'} , 
\\
(\sCurlDagger_{(n,m)} \varphi)_{A_1 \dots A_{k-1}}{}^{A'_1 \dots A'_{l+1}}
&= 
\left[\nabla^{B(A'_1} + n U^{B(A'_1} + m \bar U^{B(A'_1} \right] \varphi_{A_1 \dots A_{k-1} B}{}^{A'_2 \dots A'_{l+1})} , 
\\
(\sTwist_{(n,m)} \varphi)_{A_1 \dots A_{k+1}}{}^{A'_1 \dots A'_{l+1}}
&=
\left[\nabla_{(A_1}{}^{(A'_1} + n U_{(A_1}{}^{(A'_1}+ m \bar U_{(A_1}{}^{(A'_1} \right] \varphi_{A_2 \dots A_{k+1})}{}^{A'_2 \dots A'_{l+1})} .
\end{align}
\end{subequations}
It follows that the commutator of extended fundamental spinor operators with $n_1 = n_2, m_1 = m_2$ reduces to the commutator of the usual fundamental spinor operators \cite[\S 2.2]{ABB:symop:2014CQGra..31m5015A}. For commutators of the extended operators with unequal weights $n_1,n_2,m_1,m_2$ one simply splits into first derivatives and remainder with equal weights.

\subsection{Projection operators and the spin decomposition} \label{sec:projop}
The Killing spinor $\kappa_{AB}$ plays a central role in the geometry of Petrov type D spaces. The tensor product of $\kappa_{AB}$ with a  symmetric spinor has at most three different irreducible components. These involve either zero, one or two contractions and symmetrization. For these operations we introduce the $\mathcal{K}$-operators in
\begin{definition}
Given the Killing spinor \eqref{eq:TypeDKS}, define the operators $\mathcal{K}^i:\mathcal{S}_{k,l}\rightarrow \mathcal{S}_{k-2i+2,l}, i=0,1,2$  via
\begin{subequations} \label{eq:Kprojectors}
\begin{align}
(\mathcal{K}^0 \varphi)_{A_1\dots A_{k+2} A'_1\dots A'_l}={}&2\kappa_1^{-1} \kappa_{(A_1A_2}\varphi_{A_3\dots A_{k+2}) A'_1\dots A'_l},\\
(\mathcal{K}^1 \varphi)_{A_1\dots A_k A'_1\dots A'_l}={}&
\kappa_1^{-1}\kappa_{(A_1}{}^{F}\varphi_{A_2\dots A_k)F A'_1\dots A'_l},\\
(\mathcal{K}^2 \varphi)_{A_1\dots A_{k-2} A'_1\dots A'_l}={}&- \tfrac{1}{2}\kappa_1^{-1}\kappa^{CD} \varphi_{A_1\dots A_{k-2} CDA'_1\dots A'_l}.
\end{align}
\end{subequations}
\end{definition}
Note that the complex conjugated operators act on the primed indices in the analogous way. The action of the $\mathcal{K}$-operators does have an interpretation in terms of the resulting components with respect to a principal dyad.
\begin{example}
The ``spin raising'' operator\footnote{The name spin raising and lowering is due to the fact that multiplication and symmetrization or contraction of a spin-$\sfrak$ field with a valence-2 Killing spinor leads to a spin-$\sfrak+1$  or spin-$\sfrak-1$ field respectively, see  \citet[Sec. 6.4]{Penrose:1986fk}.} $\mathcal{K}^0$ on $\varphi_{AB} \in \SymSpinSec_{2,0}$ has components
\begin{align*}
(\mathcal{K}^0 \varphi)_{0}={}&0, &
(\mathcal{K}^0 \varphi)_{1}={}&\varphi_{0}, &
(\mathcal{K}^0 \varphi)_{2}={}&\tfrac{4}{3} \varphi_{1}, &
(\mathcal{K}^0 \varphi)_{3}={}&\varphi_{2}, &
(\mathcal{K}^0 \varphi)_{4}={}&0.
\end{align*}
The ``sign flip'' operator $\mathcal{K}^1$ on $\varphi_{ABCD} \in \SymSpinSec_{4,0}$ has components
\begin{align*}
(\mathcal{K}^1 \varphi)_{0}={}&\varphi_{0}, &
(\mathcal{K}^1 \varphi)_{1}={}&\tfrac{1}{2} \varphi_{1}, &
(\mathcal{K}^1 \varphi)_{2}={}&0, &
(\mathcal{K}^1 \varphi)_{3}={}&- \tfrac{1}{2} \varphi_{3}, &
(\mathcal{K}^1 \varphi)_{4}={}&- \varphi_{4}.
\end{align*}
The ``spin lowering'' operator $\mathcal{K}^2$ on $\varphi_{ABCD} \in \SymSpinSec_{4,0}$ has components
\begin{align*}
(\mathcal{K}^2 \varphi)_{0}={}&\varphi_{1}, &
(\mathcal{K}^2 \varphi)_{1}={}&\varphi_{2}, &
(\mathcal{K}^2 \varphi)_{2}={}&\varphi_{3}.
\end{align*}
\end{example}
\begin{definition}[Spin decomposition] \label{def:SpinDecomposition}
For any symmetric spinor $\varphi_{A_1 \dots A_{2s}}$,
\begin{itemize}
 \item with integer $s$, define $s+1$ symmetric valence $2s$ spinors $(\mathcal{P}^{i} \varphi)_{A_1 \dots A_{2s}}, i = 0 \dots s$ solving
 \begin{align}
  \varphi_{A_1 \dots A_{2s}} = \sum_{i=0}^{s} (\mathcal{P}^{i} \varphi)_{A_1 \dots A_{2s}},
 \end{align}
 with $(\mathcal{P}^{i} \varphi)_{A_1 \dots A_{2s}}$ depending only on the components $\varphi_{s+i}$ and $\varphi_{s-i}$.
 \item with half-integer $s$, define $s + \tfrac{1}{2}$ symmetric valence $2s$ spinors  $(\mathcal{P}^{i} \varphi)_{A_1 \dots A_{2s}}, i = \tfrac{1}{2} \dots s$ solving
  \begin{align}
   \varphi_{A_1 \dots A_{2s}} = \sum_{i=1/2}^{s} (\mathcal{P}^{i} \varphi)_{A_1 \dots A_{2s}},
  \end{align}
  with $(\mathcal{P}^{i} \varphi)_{A_1 \dots A_{2s}}$ depending only on the components $\varphi_{s+i}$ and $\varphi_{s-i}$.
\end{itemize}
\end{definition}

\begin{remark}
The spin decomposition can also be defined for spinors with primed indices and more generally for mixed valence. In that case the decompositions combine linearly.
\end{remark}

\begin{example}
\begin{enumerate}
 \item  For $s = 2$ the decomposition is given by
 \begin{align} \label{eq:Spin2Decomposition}
 \varphi_{ABCD}= (\mathcal{P}^{0} \varphi)_{ABCD}+(\mathcal{P}^{1} \varphi)_{ABCD}+(\mathcal{P}^{2} \varphi)_{ABCD}
 \end{align}
 and the components, written as vectors, are
 \begin{align}
 \begin{pmatrix}
 \varphi_0 \\ \varphi_1 \\ \varphi_2 \\ \varphi_3 \\ \varphi_4 
 \end{pmatrix}
 = 
  \begin{pmatrix}
  0 \\ 0 \\ \varphi_2 \\ 0 \\ 0 
  \end{pmatrix}
  +
   \begin{pmatrix}
   0 \\ \varphi_1 \\ 0 \\ \varphi_3 \\ 0 
   \end{pmatrix}
   +
    \begin{pmatrix}
    \varphi_0 \\ 0 \\ 0 \\ 0 \\ \varphi_4 
    \end{pmatrix}.
 \end{align}
In terms of the operators \eqref{eq:Kprojectors} they read
\begin{subequations}
\begin{align}
(\mathcal{P}^{0} \varphi)_{ABCD}={}&\tfrac{3}{8} (\mathcal{K}^0 \mathcal{K}^0 \mathcal{K}^2 \mathcal{K}^2 \varphi)_{ABCD}, \label{eq:P0Def}\\
(\mathcal{P}^{1} \varphi)_{ABCD}={}&(\mathcal{K}^0 \mathcal{K}^1 \mathcal{K}^1 \mathcal{K}^2 \varphi)_{ABCD},\label{eq:P1Def}\\
(\mathcal{P}^{2} \varphi)_{ABCD}={}& (\mathcal{K}^1 \mathcal{K}^1 \mathcal{K}^1 \mathcal{K}^1 \varphi)_{ABCD} - \tfrac{1}{16}  (\mathcal{K}^0 \mathcal{K}^1 \mathcal{K}^1 \mathcal{K}^2 \varphi)_{ABCD} \label{eq:P2Def}
 .
\end{align}
\end{subequations}
 \item For $s = 3/2$ on $\SymSpinSec_{3,1}$ the decomposition is given by
  \begin{align} \label{eq:Spin32Decomposition}
  \varphi_{ABCA'}= (\mathcal{P}^{1/2} \varphi)_{ABCA'}+(\mathcal{P}^{3/2} \varphi)_{ABCA'}
  \end{align}
  and the components, written as vectors, are
  \begin{align}
  \begin{pmatrix}
  \varphi_{0A'} \\ \varphi_{1A'} \\ \varphi_{2A'} \\ \varphi_{3A'}
  \end{pmatrix}
  = 
   \begin{pmatrix}
   0 \\ \varphi_{1A'} \\ \varphi_{2A'} \\ 0 
   \end{pmatrix}
   +
    \begin{pmatrix}
    \varphi_{0A'} \\ 0 \\ 0 \\ \varphi_{3A'}  
    \end{pmatrix}.
  \end{align}
 In terms of the operators \eqref{eq:Kprojectors} they read
 \begin{subequations}
\begin{align}
(\mathcal{P}^{1/2} \varphi)_{ABCA'}={}&\tfrac{3}{4} (\mathcal{K}^0 \mathcal{K}^2 \varphi)_{ABCA'}, \label{eq:P12Def}\\
(\mathcal{P}^{3/2} \varphi)_{ABCA'}={}&- \tfrac{1}{12} (\mathcal{K}^0 \mathcal{K}^2 \varphi)_{ABCA'}
 + (\mathcal{K}^1 \mathcal{K}^1 \varphi)_{ABCA'}. \label{eq:P32Def}
\end{align}
 \end{subequations}
\end{enumerate}
\end{example}
For the proof of the main theorem we need various commutator relations of the operators introduced above. The general commutators of $\mathcal{K}$-operators with extended fundamental operators for any valence are collected in appendix~\ref{app:KopComm} and here we restrict to the special cases needed for the proof. For notational convenience indices are not written explicitly.
\begin{lemma}\label{lem:Kcommutators}
On $\SymSpinSec_{k,l}$ for specific $k,l$, we have the algebraic identities
\begin{subequations}
\begin{align}
 \mathcal{K}^1 \mathcal{P}^{1}={}&\tfrac{1}{2} \mathcal{K}^0 \mathcal{K}^1 \mathcal{K}^2 &&\text{on } \SymSpinSec_{4,0}, \label{eq:K1P1phi40}\\
\mathcal{K}^1 \mathcal{K}^1 ={}&\mathrm{Id}
&&\text{on } \SymSpinSec_{1,1},\label{eq:K1K1phi11}\\
\mathcal{K}^2 \mathcal{K}^1 ={}&0
&&\text{on } \SymSpinSec_{2,0}, \label{eq:K2K1phi20}\\
 \mathcal{P}^{3/2} \mathcal{K}^0 ={}&0
 &&\text{on } \SymSpinSec_{1,1}, \label{eq:Spin32K0phi11} \\
\mathcal{K}^1 \mathcal{K}^1 \mathcal{K}^1 ={}&\mathcal{K}^1 
&&\text{on } \SymSpinSec_{2,0}, \label{eq:K1K1K1phi20} \\
\mathcal{K}^1 \mathcal{K}^1 \mathcal{K}^1 ={}&- \tfrac{2}{9} \mathcal{K}^0 \mathcal{K}^1 \mathcal{K}^2 + \mathcal{K}^1 
&&\text{on } \SymSpinSec_{3,1}, \label{eq:K1K1K1phi31}
\end{align}
\end{subequations}
and, for any integer $w$, the first order differential identities
\begin{subequations}
\begin{align}
\sCurlDagger_{(w,0)} \mathcal{K}^1 ={}&\mathcal{K}^1 \sCurlDagger_{(w,0)}
 + \tfrac{1}{2} \sTwist_{(w-4,0)} \mathcal{K}^2 
 &&\text{on } \SymSpinSec_{4,0}, \label{eq:CurlDgEK1phi40}\\
\sCurlDagger_{(w,0)} \mathcal{K}^0 ={}&\mathcal{K}^0 \sCurlDagger_{(w-1,0)} 
 - \sTwist_{(w-4,0)} \mathcal{K}^1 
 &&\text{on } \SymSpinSec_{2,0}, \label{eq:CurlDgEK0phi20}\\
\sCurlDagger_{(w,0)} \mathcal{K}^0 ={}&\mathcal{K}^0 \sCurlDagger_{(w-1,0)} 
 - \tfrac{4}{3} \sTwist_{(w-3,0)} \mathcal{K}^1
 &&\text{on } \SymSpinSec_{1,1},\label{eq:CurlDgEK0phi11}\\
\sCurlDagger_{(w,0)} \mathcal{K}^1 ={}&\mathcal{K}^1 \sCurlDagger_{(w,0)} 
 + \sTwist_{(w-2,0)} \mathcal{K}^2
 &&\text{on } \SymSpinSec_{2,0},\label{eq:CurlDgEK1phi20}\\
\sCurlDagger_{(w,0)} \mathcal{K}^2 ={}&\mathcal{K}^2 \sCurlDagger_{(w+1,0)} 
&&\text{on } \SymSpinSec_{4,0},\label{eq:CurlDgEK2phi40}\\
\sCurlDagger_{(w,0)} \mathcal{K}^0 ={}&-2 \mathcal{K}^1 \sTwist_{(w-2,0)}
&&\text{on } \SymSpinSec_{0,0}, \label{eq:CurlDgEK0phi00}\\
\sDiv_{(w,0)} \mathcal{K}^1 ={}&2 \mathcal{K}^2 \sCurl_{(w-2,0)}
&&\text{on } \SymSpinSec_{1,1}, \label{eq:DivEK1phi11}\\
\sDiv_{(w,0)} \mathcal{K}^2 ={}&\mathcal{K}^2 \sDiv_{(w+1,0)}
&&\text{on } \SymSpinSec_{3,1}, \label{eq:DivEK2phi31} \\
\mathcal{K}^1 \mathcal{P}^{3/2} \sTwist_{(w,0)} ={}&- \tfrac{1}{4} \sCurlDagger_{(w+4,0)} \mathcal{K}^0 \mathcal{K}^1 \mathcal{K}^1 
+ \tfrac{3}{4} \sTwist_{(w,0)} \mathcal{K}^1
&&\text{on } \SymSpinSec_{2,0}. \label{eq:K1Spin32TwistEphi20}
\end{align}
\end{subequations}
\end{lemma} 
\begin{proof}
For \eqref{eq:K1P1phi40} we calculate in $\SymSpinSec_{4,0}$
\begin{align*}
\mathcal{K}^1 \mathcal{P}^{1} =  \mathcal{K}^1 \mathcal{K}^0 \mathcal{K}^1 \mathcal{K}^1 \mathcal{K}^2 = \tfrac{1}{2} \mathcal{K}^0 \mathcal{K}^1 \mathcal{K}^1 \mathcal{K}^1 \mathcal{K}^2 =  \tfrac{1}{2} \mathcal{K}^0 \mathcal{K}^1 \mathcal{K}^2.
\end{align*}
The first step uses \eqref{eq:P1Def}, the second one is a commutator of $\mathcal{K}^1$ and $\mathcal{K}^0$ and the third step makes use of the fact that three sign-flips are equal to one sign-flip. For \eqref{eq:K1K1phi11} we note that $\mathcal{K}^1$  in $\SymSpinSec_{1,1}$ changes sign in two of the four components,
\begin{align*}
(\mathcal{K}^1 \varphi)_{00'}={}&\varphi_{00'},&
(\mathcal{K}^1 \varphi)_{01'}={}&\varphi_{01'},&
(\mathcal{K}^1 \varphi)_{10'}={}&- \varphi_{10'},&
(\mathcal{K}^1 \varphi)_{11'}={}&- \varphi_{11'},
\end{align*}
so $\mathcal{K}^1 \mathcal{K}^1 = \textrm{Id}$ in $\SymSpinSec_{1,1}$. Equation \eqref{eq:K2K1phi20} is true because $\mathcal{K}^1$ cancels the middle component of $\varphi_{AB}$ and $\mathcal{K}^2$ singles out that middle component.
The rest of the algebraic identities are proved analogously. 
The proof of the differential identities relies on a straightforward but tedious expansion of projectors \eqref{eq:Kprojectors} and extended fundamental operators \eqref{eq:ExtendedFundamentalOperators2}. We only calculate \eqref{eq:CurlDgEK2phi40}.
\begin{align*}
(\sCurlDagger_{(-1,0)} \mathcal{K}^2 \varphi)_{AA'} -  (\mathcal{K}^2 \sCurlDagger \varphi)_{AA'}={}&\frac{U^{B}{}_{A'} \kappa^{CD} \varphi_{ABCD}}{2 \kappa_1}
 + \frac{\kappa^{BC} (\sCurlDagger \varphi)_{ABCA'}}{2 \kappa_1}
 + (\sCurlDagger \mathcal{K}^2 \varphi)_{AA'} \nonumber \\
={}&\frac{U^{B}{}_{A'} \kappa^{CD} \varphi_{ABCD}}{2 \kappa_1}
 -  \frac{\varphi_{ABCD} (\sTwist \kappa)^{BCD}{}_{A'}}{2 \kappa_1}\nonumber\\
& + \frac{\kappa^{CD} \varphi_{ABCD} (\sTwist \kappa_1)^{B}{}_{A'}}{2 \kappa_1^2} \nonumber \\
={}&0.
\end{align*}
This identity can be generalized by replacing $\varphi \to \kappa_1^{-1-w} \varphi$ and commuting out $\kappa_1^{-1-w}$ using \eqref{eq:ExtendedFundamentalOperators1},
which results in \eqref{eq:CurlDgEK2phi40}.
The other identities are proved along the same lines. General commutators of $\mathcal{K}$ operators and fundamental spinor operators for arbitrary valence are collected in appendix~\ref{app:KopComm}.
\end{proof}

\begin{lemma}\label{lem:strangeidentities}
In $\SymSpinSec_{0,0}$ we have the identities
\begin{subequations} \label{eq:2ndOrderKCommutatorIdentities}
\begin{align}
0={}&\mathcal{K}^0 \sCurlDagger_{(-5,0)} \sTwist_{(-3,0)} 
 + \tfrac{4}{3} \sTwist_{(-4,0)} \mathcal{K}^1 \sTwist_{(-6,0)}   -  \tfrac{4}{3} \sTwist_{(-7,0)} \mathcal{K}^1 \sTwist_{(-3,0)}, \label{eq:K0CurlDgETwistEK2K2phi40}\\
0={}&\sCurlDagger_{(-1,0)} \mathcal{K}^0 \sTwist 
 + 2 \mathcal{K}^0 \sCurlDagger \sTwist 
 + 12 \mathcal{K}^1 \sTwist \sTwist  -  \tfrac{32}{3} \sTwist_{(-1,0)} \mathcal{K}^1 \sTwist_{(3/2,0)}, \label{eq:CurlDgK0Twistphi00}
\end{align}
and in $\SymSpinSec_{2,2}$ we have
\begin{align} 
0={}&- \tfrac{1}{12} \sCurlDagger_{(-1,0)} \mathcal{K}^0 \sDiv_{(4,0)}
 + \sCurlDagger_{(-1,0)} \mathcal{K}^1 \sCurl_{(1,0)} 
 -  \Psi_{2} \mathcal{K}^1 -  \mathcal{K}^1 \sCurlDagger  \sCurl \nonumber\\
  &  -  \tfrac{1}{3} \mathcal{K}^1 \sTwist \sDiv 
 + \frac{\mathcal{L}_{\xi}}{3 \kappa_1}  + \tfrac{2}{9} \sTwist_{(-1,0)} \mathcal{K}^1 \sDiv_{(1,0)}
 -  \tfrac{2}{3} \sTwist_{(-1,0)} \mathcal{K}^2 \sCurl_{(-2,0)}. \label{eq:K1CurlDgCurlphi22}
\end{align}
\end{subequations}
\end{lemma}
\begin{proof} This can be verified by expanding all operators in terms of the non-extended fundamental spinor operators. To prove  \eqref{eq:K0CurlDgETwistEK2K2phi40}, we start with $\mathcal{K}^0\sCurlDagger_{(-5,0)} \sTwist_{(-3,0)}$ and commute the $\mathcal{K}$ operator inside once, yielding,
\begin{align} \label{eq:K0CurlDgETwistEphi00Id1}
\mathcal{K}^0 \sCurlDagger_{(-5,0)} \sTwist_{(-3,0)} ={}&\sCurlDagger_{(-4,0)} \mathcal{K}^0 \sTwist_{(-3,0)}
 + \tfrac{4}{3}\sTwist_{(-7,0)} \mathcal{K}^1 \sTwist_{(-3,0)}.
\end{align}
For the first term on the RHS, we commute the $\sCurlDagger$ and $\sTwist$ through the $\mathcal{K}$ operator, via,
\begin{align}
\sCurlDagger_{(-4,0)} \mathcal{K}^0 \sTwist_{(-3,0)}
= \sCurlDagger_{(-4,0)} \sTwist_{(-4,0)} \mathcal{K}^0 
= \tfrac{2}{3}\sTwist_{(-4,0)} \sCurlDagger_{(-4,0)} \mathcal{K}^0  
=- \tfrac{4}{3}\sTwist_{(-4,0)} \mathcal{K}^1 \sTwist_{(-6,0)}.
\end{align}
which, after substitution in \eqref{eq:K0CurlDgETwistEphi00Id1} gives \eqref{eq:K0CurlDgETwistEK2K2phi40}. The other identities can be proved similarly.  Alternatively the verification can be done by expanding in components using the \emph{SpinFrames} package \citep{SpinFrames}. 
\end{proof} 

\section{Spinorial formulation of linearized gravity} \label{sec:LinearizedGravity}

In this section we review the field equations of linearized gravity for a general vacuum background and allow for sources of the linearized field. The spinor variational operator $\vartheta$ developed in \citep{BaeVal15} will be used. $\vartheta$ is an operator mapping covariant spinors of valence $(k,l)$ to covariant spinors of the same valence. Due to covariance it is independent of linearized frame or dyad transformation.

Let $\dot g_{ab}$ be a real linearized metric and $\dot g_{ABA'B'}$ the spinorial version. Observe that we make the variation with the indices down, and raise them and take traces afterwards. We define the irreducible parts of the linearized metric as
\begin{align}
G_{ABA'B'}={}&\dot g_{(AB)(A'B')}, & 
\slashed{G}_{}={}&\dot g^{C}{}_{C}{}^{C'}{}_{C'},
\end{align} 
so the decomposition into traceless and trace parts is given by
\begin{align} \label{eq:deltagIrrDec}
\dot{g}_{ABA'B'}={}&G_{ABA'B'} + \tfrac{1}{4} \slashed{G} \epsilon_{AB} \bar\epsilon_{A'B'}.
\end{align}
For the spinor variation of the irreducible parts of the curvature we get, see \citep{BaeVal15}, for a general vacuum background
\begin{subequations} \label{eq:LinGravEqs}
\begin{align}
 \vartheta \Psi_{ABCD} ={}&\tfrac{1}{2} (\sCurl \sCurl G)_{ABCD} - \tfrac{1}{4} \slashed{G}_{} \Psi_{ABCD} \label{eq:VarSPsiCDe}\\
\vartheta \Phi_{ABA'B'}={}&\tfrac{1}{2} G^{CD}{}_{A'B'} \Psi_{ABCD}
 + \tfrac{1}{2} (\sCurlDagger \sCurl G)_{ABA'B'}
 + \tfrac{1}{6} (\sTwist \sDiv G)_{ABA'B'} -  \tfrac{1}{8} (\sTwist \sTwist \slashed{G})_{ABA'B'}, \label{eq:VarSPhiCDe}\\
\vartheta \Lambda={}&- \tfrac{1}{24} (\sDiv \sDiv G)
 + \tfrac{1}{32} (\sDiv \sTwist \slashed{G}).\label{eq:VarSLambdaCDe}
\end{align}
\end{subequations} 
Note also that $\vartheta \Lambda = 0 = \vartheta \Phi_{ABA'B'}$ in the source-free case, i.e. when $\dot g_{ABA'B'}$ is a solution to the linearized vacuum Einstein equations.
 
For later use, see lemma~\ref{lem:Mcurvature} below, we introduce the notation $\vartheta \Psi[\varphi,\slashed\varphi]_{ABCD}$ for the linearized Weyl curvature operator acting on a symmetric tensor field with irreducible parts $\varphi_{ABA'B'}, \slashed\varphi$ (and analogously the other curvature operators). In case the field is given by $G_{ABA'B'},\slashed G$ or 0, we suppress the additional argument. 
It will also be convenient to introduce
\begin{align} 
\phi_{ABCD}={}&\tfrac{1}{2} (\sCurl \sCurl G)_{ABCD} = \vartheta \Psi_{ABCD} + \tfrac{1}{4} \slashed{G}_{} \Psi_{ABCD}. \label{eq:phidef} 
\end{align}
as a modification of the varied Weyl spinor $\vartheta\Psi_{ABCD}$\footnote{In a type D principal frame this modification only affects the middle component.}. In some equations it will be more convenient to use $\phi_{ABCD}$ and in others to use $\vartheta\Psi_{ABCD}$.

As a consequence of \eqref{eq:LinGravEqs} we derive the linearized Bianchi identity. Strictly speaking it is only the variation of \eqref{eq:Bianchi1}, since the variation of \eqref{eq:Bianchi2} is not needed for the result of this paper.
\begin{lemma}
For a general vacuum background the modified Weyl spinor $\phi_{ABCD}$ defined in \eqref{eq:phidef} satisfies
\begin{align}
(\sCurlDagger \phi)_{ABCA'}={}&(\sCurl \vartheta \Phi)_{ABCA'}
 + \tfrac{1}{2} \Psi_{ABCD} (\sDiv G)^{D}{}_{A'}
 -  \tfrac{3}{2} \Psi_{(AB}{}^{DF}(\sCurl G)_{C)DFA'}\nonumber\\
& -  \tfrac{1}{8} \Psi_{ABCD} (\sTwist \slashed{G})^{D}{}_{A'}
 + \tfrac{1}{2} G^{DF}{}_{A'}{}^{B'} (\sTwist \Psi)_{ABCDFB'}. \label{eq:LinBianchi1}
\end{align}
Restricting to a type D background this simplifies to
\begin{align}
(\sCurlDagger \phi)_{ABCA'}={}&(\sCurl \vartheta \Phi)_{ABCA'}
 -  \tfrac{3}{2} \Psi_{2} (\sCurl_{(1,0)} G)_{ABCA'}
 -  \tfrac{3}{8} \Psi_{2} (\mathcal{K}^0 \mathcal{K}^1 \sDiv_{(4,0)} G)_{ABCA'}\nonumber\\
& + \tfrac{3}{32} \Psi_{2} (\mathcal{K}^0 \mathcal{K}^1 \sTwist \slashed{G})_{ABCA'}
 + \tfrac{9}{4} \Psi_{2} (\mathcal{K}^0 \mathcal{K}^2 \sCurl_{(1,0)} G)_{ABCA'}.\label{eq:LinBianchi2}
\end{align}
\end{lemma}
\begin{proof}
We apply $\sCurl$ on \eqref{eq:VarSPhiCDe}, commute $\sCurl \sCurlDagger$, use \eqref{eq:phidef} to get
\begin{align}
(\sCurl \vartheta \Phi)_{ABCA'}
={}&\tfrac{1}{2} (\sCurl \sCurlDagger \sCurl G)_{ABCA'}
 + \tfrac{1}{6} (\sCurl \sTwist \sDiv G)_{ABCA'}
 -  \tfrac{1}{8} (\sCurl \sTwist \sTwist \slashed{G})_{ABCA'}\nonumber\\
& -  \tfrac{1}{3} \Psi_{ABCD} (\sDiv G)^{D}{}_{A'}
 + \tfrac{1}{2} \Psi_{(AB}{}^{DF}(\sCurl G)_{C)DFA'}
  -  \tfrac{1}{2} G^{DF}{}_{A'}{}^{B'} (\sTwist \Psi)_{ABCDFB'}\nonumber\\
={}&\tfrac{1}{6} (\sCurl \sTwist \sDiv G)_{ABCA'}
 -  \tfrac{1}{8} (\sCurl \sTwist \sTwist \slashed{G})_{ABCA'}
 + (\sCurlDagger \phi)_{ABCA'}
  -  \tfrac{1}{3} \Psi_{ABCD} (\sDiv G)^{D}{}_{A'}\nonumber\\
& + \tfrac{3}{2} \Psi_{(AB}{}^{DF}(\sCurl G)_{C)DFA'}
 -  \tfrac{1}{2} G^{DF}{}_{A'}{}^{B'} (\sTwist \Psi)_{ABCDFB'}
 -  \tfrac{1}{8} (\sTwist \sDiv \sCurl G)_{ABCA'} .
\end{align}
We then commute the $\sCurl \sTwist$ operators, and in the last step we commute $\sDiv \sCurl$ and use $\sCurl \sTwist = 0$, valid on $\SymSpinSec_{0,0}$, to get 
\begin{align}
(\sCurl \vartheta \Phi)_{ABCA'}
={}&(\sCurlDagger \phi)_{ABCA'}
 -  \tfrac{1}{2} \Psi_{ABCD} (\sDiv G)^{D}{}_{A'}
 + \tfrac{3}{2} \Psi_{(AB}{}^{DF}(\sCurl G)_{C)DFA'}\nonumber\\
& + \tfrac{1}{8} \Psi_{ABCD} (\sTwist \slashed{G})^{D}{}_{A'}
 -  \tfrac{1}{2} G^{DF}{}_{A'}{}^{B'} (\sTwist \Psi)_{ABCDFB'}\nonumber\\
& + \tfrac{1}{12} (\sTwist \sCurl \sDiv G)_{ABCA'}
 -  \tfrac{1}{16} (\sTwist \sCurl \sTwist \slashed{G})_{ABCA'}
  -  \tfrac{1}{8} (\sTwist \sDiv \sCurl G)_{ABCA'}\nonumber\\
={}&(\sCurlDagger \phi)_{ABCA'}
 -  \tfrac{1}{2} \Psi_{ABCD} (\sDiv G)^{D}{}_{A'}
 + \tfrac{3}{2} \Psi_{(AB}{}^{DF}(\sCurl G)_{C)DFA'}\nonumber\\
& + \tfrac{1}{8} \Psi_{ABCD} (\sTwist \slashed{G})^{D}{}_{A'}
 -  \tfrac{1}{2} G^{DF}{}_{A'}{}^{B'} (\sTwist \Psi)_{ABCDFB'}.
\end{align}
This gives \eqref{eq:LinBianchi1}. On a type D spacetime, we can use \eqref{eq:TwistPsiCDe} and \eqref{eq:PsiCDeTokappa}. The resulting $U_{AA'}$ spinors  can be incorporated as extended indices, and the $\kappa_{AB}$ spinors can then be rewritten in terms the $\mathcal{K}$ operators to get \eqref{eq:LinBianchi2}.
\end{proof}
Note that on a Minkowski background and without sources, the right-hand side of \eqref{eq:LinBianchi1} vanishes and the linearized Bianchi identity reduces to the spin-2 equation. The linearized Bianchi identity \eqref{eq:LinBianchi2} is of fundamental importance and next we derive some differential identities for it which are needed for the main theorem. The following variable appears naturally.
\begin{definition}\label{example:spin2proj} 
Define the symmetric spinor $\widehat{\phi}_{ABCD}$ as the rescaled, sign-flipped and spin-2 projected linearized Weyl spinor, 
\begin{align} \label{eq:phihatDef}
\widehat\phi_{ABCD} = \kappa_1^4 (\mathcal{K}^1 \mathcal{P}^{2} \vartheta\Psi)_{ABCD}.
\end{align} 
\end{definition} 
The components of $\widehat\phi_{ABCD}$ in a principal dyad are 
$$
\begin{pmatrix} \widehat\phi_0 \\ \widehat\phi_1 \\ \widehat\phi_2 \\ \widehat\phi_3 \\ \widehat\phi_4 \end{pmatrix} =
\begin{pmatrix} \kappa_1^4 \vartheta\Psi_0 \\ 0 \\ 0 \\ 0 \\ - \kappa_1^4 \vartheta\Psi_4 \end{pmatrix}. 
$$

\begin{corollary}\label{cor:cD4-id}
An alternative form of the linearized Bianchi identity \eqref{eq:LinBianchi2}, involving the variable $\widehat\phi_{ABCD}$ defined in \eqref{eq:phihatDef}, is given by
\begin{align} \label{eq:CurlDgEphi40t}
(\mathcal{P}^{3/2} \sCurlDagger_{(4,0)} \widehat\phi)_{ABCA'}
 ={}& - \bigl(\kappa_1^4\mathcal{K}^1 \mathcal{P}^{3/2} \sCurlDagger  \mathcal{P}^{1} \phi\bigr)_{ABCA'}  - \tfrac{3}{2} \bigl(\kappa_1^4\Psi_{2}\mathcal{K}^1 \mathcal{P}^{3/2} \sCurl_{(1,0)} G\bigr)_{ABCA'} \nonumber \\
 {}& + \bigl(\kappa_1^4\mathcal{K}^1 \mathcal{P}^{3/2} \sCurl \vartheta \Phi\bigr)_{ABCA'}.
\end{align}
\end{corollary}
\begin{proof}
Applying the operator $\mathcal{K}^1 \mathcal{P}^{3/2} \kappa_1^4\sCurlDagger $ to the spin decomposition \eqref{eq:Spin2Decomposition} with $\varphi=\vartheta\Psi$ and using \eqref{eq:CurlDgEK1phi40} gives the identity
\begin{align}
(\sCurlDagger_{(4,0)} \widehat\phi)_{ABCA'}={}&- \bigl(\mathcal{K}^1 \mathcal{P}^{3/2} \sCurlDagger_{(4,0)} (\kappa_1^4\mathcal{P}^{1} \phi)\bigr)_{ABCA'}
 + \bigl(\mathcal{K}^1 \mathcal{P}^{3/2} (\kappa_1^4 \sCurlDagger \phi)\bigr)_{ABCA'}.
\end{align}
The result follows from $\mathcal{P}^{3/2} \mathcal{K}^0 = 0$, valid in $\SymSpinSec_{1,1}$, together with \eqref{eq:LinBianchi2}.
\end{proof}

\begin{corollary}[Covariant TME] \label{cor:covTME}
The covariant form of the spin-2 Teukolsky master equation with source on a vacuum type D background is given by
\begin{align} \label{eq:covTME}
(\sCurl \sCurlDagger_{(4,0)} \widehat\phi)_{ABCD}={}&-3 \Psi_{2} \widehat\phi_{ABCD}
 + \kappa_1^4 (\mathcal{P}^{2} \mathcal{K}^1 \sCurl_{(-4,0)} \sCurl \vartheta \Phi)_{ABCD}.
\end{align}
\end{corollary}
\begin{proof}
Apply the operator $\mathcal{P}^2 \sCurl$ to \eqref{eq:CurlDgEphi40t} and use
\begin{subequations}
\begin{align}
\bigl(\mathcal{P}^{2} \sCurl \mathcal{K}^1 \mathcal{P}^{3/2} \sCurlDagger_{(4,0)} (\kappa_1^4\mathcal{P}^{1} \phi)\bigr)_{ABCD}={}&0, \\
\bigl(\mathcal{P}^{2} \sCurl \mathcal{K}^1 \mathcal{P}^{3/2} (\kappa_1^4\sCurl \vartheta \Phi)\bigr)_{ABCD}={}&\kappa_1^4 (\mathcal{P}^{2} \mathcal{K}^1 \sCurl_{(-4,0)} \sCurl \vartheta \Phi)_{ABCD}, \\
\bigl(\mathcal{P}^{2} \sCurl  \mathcal{K}^1 \mathcal{P}^{3/2} (\kappa_1^4\Psi_{2}\sCurl_{(1,0)} G)\bigr)_{ABCD}={}&\Psi_{2} \kappa_1^4 (\mathcal{P}^{2} \mathcal{K}^1 \sCurl \sCurl G)_{ABCD},
\end{align}
\end{subequations} 
which follows from repeatedly commuting $\mathcal{K}$ operators and fundamental spinor operators and using the rules given in appendix~\ref{app:KopComm}.
\end{proof}

\section{Main theorem}\label{sec:MainTheorem} 
We shall now prove our main theorem. The following is the detailed statement of Theorem~\ref{thm:MainThmTensorVersionIntro} including source terms.
\begin{theorem}\label{thm:mainthm} 
Let $\dot{g}_{ABA'B'}$ be a real solution to the linearized Einstein equations with linearized Weyl curvature $\vartheta\Psi_{ABCD}$ and linearized source $\vartheta\Phi_{ABA'B'}$ on a vacuum background of Petrov type D. Furthermore, let $\widehat\phi_{ABCD} = \kappa_1^4 (\mathcal{K}^1 \mathcal{P}^{2} \vartheta\Psi)_{ABCD}$ be the modified linearized Weyl spinor and let
\begin{align}
\mathcal{M}_{ABA'B'}={}& (\sCurlDagger \sCurlDagger_{(4,0)}\widehat{\phi})_{ABA'B'} . \label{eq:Mabdef:spinor}
\end{align}

Then we have 
\begin{align}\label{eq:CurlDgCurlDgEK1Spin2ToLie}
\mathcal{M}_{ABA'B'}
={}&\tfrac{1}{2} \nabla_{AA'}\mathcal{A}_{BB'}
 + \tfrac{1}{2} \nabla_{BB'}\mathcal{A}_{AA'}
 + \tfrac{1}{2} \Psi_{2} \kappa_1^3 (\mathcal{L}_{\xi}\dot g)_{ABA'B'}  + (\mathcal{N} \vartheta\Phi)_{ABA'B'},
\end{align}
where  the complex one form $\mathcal{A}_{AA'}$ and the source term $(\mathcal{N} \vartheta\Phi)_{ABA'B'}$ are given by
\begin{subequations} 
\begin{align} 
  \mathcal{A}_{AA'}={}&- \tfrac{1}{4} \slashed{G} \Psi_{2} \kappa_1^3 \xi_{AA'}
   -  \tfrac{1}{2} \Psi_{2} \kappa_1^3 \xi^{BB'} (\mathcal{K}^0 \mathcal{K}^2 G)_{ABA'B'}
   + \tfrac{2}{3} \kappa_1^3 \xi^{B}{}_{A'} (\mathcal{K}^1 \mathcal{K}^2 \vartheta \Psi)_{AB}\nonumber\\
  & + (\mathcal{K}^1 \sTwist \mathcal{K}^2 \mathcal{K}^2 \kappa_1^4\vartheta \Psi)_{AA'}
   + \tfrac{2}{3} \kappa_1^4 (\mathcal{K}^2 \sCurl \vartheta \Phi)_{AA'}, \label{eq:AA11Def} \\
(\mathcal{N} \vartheta\Phi)_{ABA'B'} ={}& - \bigl(\sCurlDagger (\kappa_1^4 \mathcal{K}^1 \mathcal{P}^{1/2} \sCurl \vartheta \Phi)\bigr)_{ABA'B'}
       + \bigl(\sCurlDagger (\kappa_1^4\mathcal{K}^1 \mathcal{P}^{3/2} \sCurl \vartheta \Phi)\bigr)_{ABA'B'}
       - 3 \Psi_{2} \kappa_1^4 (\mathcal{K}^1 \vartheta \Phi)_{ABA'B'}. \label{eq:NVarSPhiDef}
\end{align}
\end{subequations} 
\end{theorem} 
Before proving the theorem we collect some algebraic and differential identities for $\mathcal{A}_{AA'}$.
\begin{lemma} \label{lem:DivA}
The one-form $\mathcal{A}_{AA'}$ defined in \eqref{eq:AA11Def} has the following properties: 
\begin{subequations} 
 \begin{align} \label{eq:AU}
 \mathcal{A}^{AA'} U_{AA'}={}&- \tfrac{1}{3} \kappa_1^3(\mathcal{L}_{\xi}\mathcal{K}^2 \mathcal{K}^2 \vartheta \Psi) + \tfrac{2}{3} \kappa_{1}{}^4 U^{AA'} (\mathcal{K}^2 \sCurl \vartheta \Phi)_{AA'}\\
 (\sDiv \mathcal{A})_{}={}&- \tfrac{1}{2} \Psi_{2} \kappa_1^3 (\mathcal{L}_{\xi}\slashed{G}) \label{eq:divA} \\
(\sCurl_{(2,0)} \mathcal{K}^1 \mathcal{A})_{AB}={}&
 -  \tfrac{1}{8} \Psi_{2} \kappa_{1}{}^3 (\mathcal{L}_{\xi}\mathcal{K}^0 \slashed{G})_{AB}
  -  \tfrac{2}{3} \kappa_{1}{}^3 (\mathcal{L}_{\xi}\mathcal{K}^2 \vartheta \Psi)_{AB}\nonumber\\
&+\tfrac{2}{3} \kappa_{1}{}^4 (\sCurl_{(-2,0)} \mathcal{K}^1 \mathcal{K}^2 \sCurl \vartheta \Phi)_{AB}
 -  \tfrac{2}{3} \kappa_{1}{}^3 \xi_{(A}{}^{A'}(\mathcal{K}^2 \sCurl \vartheta \Phi)_{B)A'}. \label{eq:curlK1A}
 \end{align}
\end{subequations}  
\end{lemma}
\begin{proof}
Equation \eqref{eq:AU} can be verified by a direct calculation using \eqref{eq:UU11Def}. 
To prove \eqref{eq:divA}, we make use of the form of $\mathcal{A}_{AA'}$ given in equation \eqref{eq:AA11Eq1} below,
\begin{align}
\mathcal{A}_{AA'}={}&- \tfrac{1}{2} \bigl(\sCurlDagger \mathcal{K}^0 \mathcal{K}^2 \mathcal{K}^2 (\kappa_1^4\phi)\bigr)_{AA'}
 + \tfrac{2}{3} \bigl(\sCurlDagger  \mathcal{K}^2 (\kappa_1^4\phi)\bigr)_{AA'}
 + \tfrac{1}{3} \bigl(\mathcal{K}^1 \sDiv_{(2,0)} (\kappa_1^4\Psi_{2}G)\bigr)_{AA'}\nonumber\\
& -  \tfrac{1}{3} \bigl(\mathcal{K}^1 \sTwist_{(5/2,0)} (\kappa_1^4\Psi_{2}\slashed{G})\bigr)_{AA'}
 -  \bigl(\mathcal{K}^2 \sCurl_{(-1,0)} (\kappa_1^4\Psi_{2}G)\bigr)_{AA'}.
\end{align}
Applying $\sDiv$ to this equation and using the commutator relation $\sDiv \sCurlDagger = 0$ valid on $\SymSpinSec_{2,0}$ gives
\begin{align}
(\sDiv  \mathcal{A})={}&\tfrac{1}{3} \bigl(\sDiv  \mathcal{K}^1 \sDiv_{(2,0)} (\kappa_1^4\Psi_{2}G)\bigr)
 -  \tfrac{1}{3} \bigl(\sDiv \mathcal{K}^1 \sTwist_{(5/2,0)} (\kappa_1^4\Psi_{2}\slashed{G})\bigr)
 -  \bigl(\sDiv \mathcal{K}^2 \sCurl_{(-1,0)} (\kappa_1^4\Psi_{2}G)\bigr).
\end{align}
Using \eqref{eq:DivEK1phi11} on the first two terms, and \eqref{eq:DivEK2phi31} on the last term gives
\begin{align}
(\sDiv \mathcal{A})={}&\tfrac{2}{3} \bigl(\mathcal{K}^2 \sCurl_{(-2,0)} \sDiv_{(2,0)} (\kappa_1^4\Psi_{2}G)\bigr)
 -  \tfrac{2}{3} \bigl(\mathcal{K}^2 \sCurl_{(-2,0)} \sTwist_{(5/2,0)} (\kappa_1^4\Psi_{2}\slashed{G})\bigr)
 -  \bigl(\mathcal{K}^2 \sDiv_{(1,0)} \sCurl_{(-1,0)} (\kappa_1^4\Psi_{2}G)\bigr).
\end{align}
The first and the last term cancel due to the general identity on $\SymSpinSec_{2,2}$
\begin{align}
\mathcal{K}^2 \sDiv_{(1,0)} \sCurl_{(-1,0)} ={}&\tfrac{2}{3} \mathcal{K}^2 \sCurl_{(-2,0)} \sDiv_{(2,0)} .
\end{align}
This identity can be proven by expanding the extended indices and commuting the derivatives. In the same way we can also prove the identity 
\begin{align}
\mathcal{K}^2 \sCurl_{(v,0)} \sTwist_{(w,0)} ={}&\frac{1}{6 \kappa_1}(w - v) \mathcal{L}_{\xi},
\end{align}
on $\SymSpinSec_{0,0}$ for arbitrary constants $v$ and $w$. This finally gives \eqref{eq:divA}, where we in the last step commuted $\Psi_{2} \kappa_1^3$ through the Lie derivative.

To prove \eqref{eq:curlK1A}, we first note the following relation that is a consequence of the linearized Bianchi identities
\begin{align}
(\sCurlDagger \mathcal{K}^2 \phi)_{AA'}={}&\tfrac{3}{2} \Psi_{2} (\sCurl_{(2,0)} \mathcal{K}^2 G)_{AA'}
 + \tfrac{1}{8} \Psi_{2} (\mathcal{K}^1 \sTwist \slashed{G})_{AA'}
 + U^{B}{}_{A'} (\mathcal{K}^2 \phi)_{AB}
 + (\mathcal{K}^2 \sCurl \vartheta \Phi)_{AA'}.\label{eq:CurlDgK2phi}
\end{align}
Applying a $\mathcal{K}^1$ to \eqref{eq:AA11Def} gives after algebraic manipulations
\begin{align}
(\mathcal{K}^1 \mathcal{A})_{AA'}={}&- \tfrac{1}{4} \slashed{G} \Psi_{2} U_{AA'} \kappa_{1}{}^4
 + \tfrac{2}{3} \kappa_{1}{}^4 (\mathcal{K}^1 \mathcal{K}^2 \sCurl \vartheta \Phi)_{AA'}
 + \Psi_{2} \kappa_{1}{}^3 \xi_{A}{}^{B'} (\mathcal{K}^2 G)_{A'B'}\nonumber\\
& + \tfrac{2}{3} \kappa_{1}{}^3 \xi^{B}{}_{A'} (\mathcal{K}^2 \phi)_{AB}
 + (\sTwist_{(2,0)} \mathcal{K}^2 \mathcal{K}^2 \kappa_{1}{}^4\phi)_{AA'}
 -  \tfrac{1}{4} (\sTwist_{(2,0)} \Psi_{2}\kappa_{1}{}^4\slashed{G})_{AA'}.
\end{align}
Applying $\sCurl_{(2,0)}$ to this equation, using the commutator relation $\sCurl_{(2,0)}\sTwist_{(2,0)}=0$, valid on $\SymSpinSec_{0,0}$, and translating $\xi^{AA'}\sTwist$ terms to Lie derivatives, we end up with
\begin{align}
(\sCurl_{(2,0)} \mathcal{K}^1 \mathcal{A})_{AB}={}&\tfrac{2}{3} (\sCurl_{(2,0)} \kappa_{1}{}^4\mathcal{K}^1 \mathcal
{K}^2 \sCurl \vartheta \Phi)_{AB}
 -  \tfrac{1}{24} \Psi_{2} \kappa_{1}{}^3 (\mathcal{K}^0 \mathcal{L}_{\xi}\slashed{G})_{AB}
  + \tfrac{1}{9} \kappa_{1}{}^2 \xi_{CA'} \xi^{CA'} (\mathcal{K}^1 \mathcal{K}^2 \phi)_{AB}\nonumber\\
&
 -  \tfrac{2}{3} \kappa_{1}{}^3 (\mathcal{L}_{\xi}\mathcal{K}^2 \phi)_{AB}
  -  \tfrac{2}{3} \kappa_{1}{}^3 \xi_{(A}{}^{A'}(\sCurlDagger \mathcal{K}^2 \phi)_{B)A'}
 + \Psi_{2} \kappa_{1}{}^3 \xi_{(A}{}^{A'}(\sCurl_{(2,0)} \mathcal{K}^2 G)_{B)A'}\nonumber\\
& + \tfrac{1}{12} \Psi_{2} \kappa_{1}{}^3 \xi_{(A}{}^{A'}(\mathcal{K}^1 \sTwist \slashed{G})_{B)A'}.
\end{align}
Using \eqref{eq:CurlDgK2phi} to eliminate the $\sCurlDagger \mathcal{K}^2 \phi$ terms yields
\begin{align}
(\sCurl_{(2,0)} \mathcal{K}^1 \mathcal{A})_{AB}={}&
 -  \tfrac{1}{24} \Psi_{2} \kappa_{1}{}^3 (\mathcal{K}^0 \mathcal{L}_{\xi}\slashed{G})_{AB}
  -  \tfrac{2}{3} \kappa_{1}{}^3 (\mathcal{L}_{\xi}\mathcal{K}^2 \phi)_{AB}\nonumber\\
  & + \tfrac{2}{3} \kappa_{1}{}^4 (\sCurl_{(-2,0)} \mathcal{K}^1 \mathcal{K}^2 \sCurl \vartheta \Phi)_{AB}
 -  \tfrac{2}{3} \kappa_{1}{}^3 \xi_{(A}{}^{A'}(\mathcal{K}^2 \sCurl \vartheta \Phi)_{B)A'}.
\end{align}
Inserting \eqref{eq:phidef} leads to the result.
\end{proof}

\begin{proof}[Proof of theorem \ref{thm:mainthm}]
Apply the operator $\sCurlDagger$ to \eqref{eq:CurlDgEphi40t} and moving out the scalars $\Psi_2$ and $\kappa_1$ we get
\begin{align}\label{eq:prf1}
(\sCurlDagger \sCurlDagger_{(4,0)} \widehat\phi)_{ABA'B'}={}&\kappa_1^4 (\sCurlDagger_{(-4,0)} \mathcal{K}^1 \mathcal{P}^{3/2} \sCurl \vartheta \Phi)_{ABA'B'}
 -  \kappa_1^4 (\sCurlDagger_{(-4,0)} \mathcal{K}^1 \mathcal{P}^{3/2} \sCurlDagger \mathcal{P}^{1} \phi)_{ABA'B'}\nonumber\\
& -  \tfrac{3}{2} \Psi_{2} \kappa_1^4 (\sCurlDagger_{(-1,0)} \mathcal{K}^1 \mathcal{P}^{3/2} \sCurl_{(1,0)} G)_{ABA'B'}.
\end{align}
The second term can be rewritten by expanding the spin-1 projector according to \eqref{eq:P1Def} and commuting out the $\mathcal{K}^0$ using \eqref{eq:CurlDgEK0phi20} together with \eqref{eq:Spin32K0phi11},\eqref{eq:K1K1K1phi20}, 
\begin{align} \label{eq:prf2}
\kappa_1^4 (\sCurlDagger_{(-4,0)} \mathcal{K}^1 \mathcal{P}^{3/2} \sCurlDagger  \mathcal{P}^{1} \phi)_{ABA'B'}={}&- \kappa_1^4 (\sCurlDagger_{(-4,0)} \mathcal{K}^1 \mathcal{P}^{3/2} \sTwist_{(-4,0)} \mathcal{K}^1 \mathcal{K}^2 \phi)_{ABA'B'} \nonumber \\
={}&\tfrac{1}{4} \kappa_1^4 (\sCurlDagger_{(-4,0)} \mathcal{K}^0 \sCurlDagger_{(-1,0)} \mathcal{K}^1 \mathcal{K}^2 \phi)_{ABA'B'}
 -  \tfrac{2}{3} \kappa_1^4 (\sTwist_{(-4,0)} \sCurlDagger_{(-4,0)} \mathcal{K}^1 \mathcal{K}^1 \mathcal{K}^2 \phi)_{ABA'B'}.
\end{align}
In the second step \eqref{eq:K1Spin32TwistEphi20} and \eqref{eq:CurlDgEK0phi20} with \eqref{eq:K1K1K1phi20} and a commutator is used. To commute out the $\mathcal{K}^1 \mathcal{K}^2 $ in the first term, we first use \eqref{eq:CurlDgEK1phi20} and \eqref{eq:CurlDgEK2phi40} to get
\begin{align} \label{eq:prf3}
(\sCurlDagger_{(-4,0)} \mathcal{K}^0 \sCurlDagger_{(-1,0)} \mathcal{K}^1 \mathcal{K}^2 \phi)_{ABA'B'}={}&(\sCurlDagger_{(-4,0)} \mathcal{K}^0 \mathcal{K}^1 \mathcal{K}^2 \sCurlDagger \phi)_{ABA'B'}
 + (\sCurlDagger_{(-4,0)} \mathcal{K}^0 \sTwist_{(-3,0)} \mathcal{K}^2 \mathcal{K}^2 \phi)_{ABA'B'}, \nonumber \\
={}&(\sCurlDagger_{(-4,0)} \mathcal{K}^0 \mathcal{K}^1 \mathcal{K}^2 \sCurlDagger \phi)_{ABA'B'}
 -  \tfrac{4}{3} (\sTwist_{(-4,0)} \mathcal{K}^1 \sTwist_{(-6,0)} \mathcal{K}^2 \mathcal{K}^2 \phi)_{ABA'B'}.
\end{align}
In the second step \eqref{eq:CurlDgEK0phi11} is used together with \eqref{eq:K0CurlDgETwistEK2K2phi40}. Using \eqref{eq:prf3} in \eqref{eq:prf2} and the linearized Bianchi identity \eqref{eq:LinBianchi2} in the first term of \eqref{eq:prf3} yields
\begin{align} \label{eq:prf4}
\kappa_1^4 (\sCurlDagger_{(-4,0)} \mathcal{K}^1 \mathcal{P}^{3/2} \sCurlDagger  \mathcal{P}^{1} \phi)_{ABA'B'}={}&\tfrac{1}{4} \kappa_1^4 (\sCurlDagger_{(-4,0)} \mathcal{K}^0 \mathcal{K}^1 \mathcal{K}^2 \sCurl \vartheta \Phi)_{ABA'B'}
 -  \tfrac{1}{8} \kappa_1^4 \bigl(\sCurlDagger_{(-4,0)} \mathcal{K}^0 \mathcal{K}^1 (\Psi_{2}\mathcal{K}^1 \sDiv_{(4,0)} G)\bigr)_{ABA'B'}\nonumber\\
& + \tfrac{1}{32} \kappa_1^4 \bigl(\sCurlDagger_{(-4,0)} \mathcal{K}^0 \mathcal{K}^1 (\Psi_{2}\mathcal{K}^1 \sTwist \slashed{G})\bigr)_{ABA'B'}\nonumber\\
& + \tfrac{3}{8} \kappa_1^4 \bigl(\sCurlDagger_{(-4,0)} \mathcal{K}^0 \mathcal{K}^1 (\Psi_{2}\mathcal{K}^2 \sCurl_{(1,0)} G)\bigr)_{ABA'B'}\nonumber\\
& -  \tfrac{2}{3} \kappa_1^4 (\sTwist_{(-4,0)} \sCurlDagger_{(-4,0)} \mathcal{K}^1 \mathcal{K}^1 \mathcal{K}^2 \phi)_{ABA'B'}\nonumber\\
& -  \tfrac{1}{3} \kappa_1^4 (\sTwist_{(-4,0)} \mathcal{K}^1 \sTwist_{(-6,0)} \mathcal{K}^2 \mathcal{K}^2 \phi)_{ABA'B'}.
\end{align}
The second and third term on the right-hand side can be simplified further using \eqref{eq:K1K1phi11}. Using \eqref{eq:prf4} in \eqref{eq:prf1} and expanding the spin decomposition in the last term of \eqref{eq:prf1} leads to
\begin{align} \label{eq:prf5}
(\sCurlDagger \sCurlDagger_{(4,0)} \widehat\phi)_{ABA'B'}={}&\bigl(\sCurlDagger \mathcal{K}^1 \mathcal{P}^{3/2} (\kappa_1^4\sCurl \vartheta \Phi)\bigr)_{ABA'B'} 
 + \tfrac{1}{8} \Psi_{2} \kappa_1^4 (\sCurlDagger_{,{-1}} \mathcal{K}^0 \sDiv_{(4,0)} G)_{ABA'B'}\nonumber\\
& -  \tfrac{1}{4} \kappa_1^4 (\sCurlDagger_{(-4,0)} \mathcal{K}^0 \mathcal{K}^1 \mathcal{K}^2 \sCurl \vartheta \Phi)_{ABA'B'}
 -  \tfrac{1}{3} \Psi_{2} \kappa_1^4 (\sCurlDagger_{(-1,0)} \mathcal{K}^0 \mathcal{K}^1 \mathcal{K}^2 \sCurl_{(1,0)} G)_{ABA'B'}\nonumber\\
& -  \tfrac{1}{32} \Psi_{2} \kappa_1^4 (\sCurlDagger_{(-1,0)} \mathcal{K}^0 \sTwist \slashed{G})_{ABA'B'}
 -  \tfrac{3}{2} \Psi_{2} \kappa_1^4 (\sCurlDagger_{(-1,0)} \mathcal{K}^1 \mathcal{K}^1 \mathcal{K}^1 \sCurl_{(1,0)} G)_{ABA'B'}\nonumber\\
& + \tfrac{2}{3} \kappa_1^4 (\sTwist_{(-4,0)} \sCurlDagger_{(-4,0)} \mathcal{K}^1 \mathcal{K}^1 \mathcal{K}^2 \phi)_{ABA'B'}
 + \tfrac{1}{3} \kappa_1^4 (\sTwist_{(-4,0)} \mathcal{K}^1 \sTwist_{(-6,0)} \mathcal{K}^2 \mathcal{K}^2 \phi)_{ABA'B'}.
\end{align}
The fourth and sixth term on the right-hand side can be combined via \eqref{eq:K1K1K1phi31}. Defining the complex vector field
\begin{subequations} 
\begin{align}
\mathcal{A}_{AA'}={}&- \tfrac{1}{2} \kappa_1^4 (\sCurlDagger_{(-4,0)} \mathcal{K}^0 \mathcal{K}^2 \mathcal{K}^2 \phi)_{AA'}
 + \tfrac{2}{3} \kappa_1^4 (\sCurlDagger_{(-4,0)} \mathcal{K}^2 \phi)_{AA'}
 + \tfrac{1}{3} \Psi_{2} \kappa_1^4 (\mathcal{K}^1 \sDiv_{(1,0)} G)_{AA'}\nonumber\\
& -  \tfrac{1}{3} \Psi_{2} \kappa_1^4 (\mathcal{K}^1 \sTwist_{(3/2,0)} \slashed{G})_{AA'}
 -  \Psi_{2} \kappa_1^4 (\mathcal{K}^2 \sCurl_{(-2,0)} G)_{AA'}\label{eq:AA11Eq1} \\
={}&  \bigl(\mathcal{K}^1 \sTwist \mathcal{K}^2 \mathcal{K}^2 (\kappa_1^4\phi)\bigr)_{AA'}
 -  \tfrac{1}{4} \Psi_{2} \kappa_1^4 (\mathcal{K}^1 \sTwist_{(2,0)} \slashed{G})_{AA'}
 + \tfrac{2}{3} \kappa_1^4 (\mathcal{K}^2 \sCurl \vartheta \Phi)_{AA'}\nonumber\\
&- \tfrac{1}{2} \Psi_{2} \kappa_1^3 \xi^{BB'} (\mathcal{K}^0 \mathcal{K}^2 G)_{ABA'B'}
 + \tfrac{2}{3} \kappa_1^3 \xi^{B}{}_{A'} (\mathcal{K}^1 \mathcal{K}^2 \phi)_{AB} \label{eq:AA11Eq2}
\end{align}
\end{subequations} 
(for the second version we used $\mathcal{K}^2$ applied to the linearized Bianchi identity \eqref{eq:LinBianchi2}) we find
\begin{align}
(\sCurlDagger \sCurlDagger_{(4,0)} \widehat\phi)_{ABA'B'}={}&\bigl(\sCurlDagger \mathcal{K}^1 \mathcal{P}^{3/2} (\kappa_1^4\sCurl \vartheta \Phi)\bigr)_{ABA'B'}  - 3 \Psi_{2} \kappa_1^4 (\mathcal{K}^1 \vartheta \Phi)_{ABA'B'}  -  \tfrac{1}{4} \kappa_1^4 (\sCurlDagger_{(-4,0)} \mathcal{K}^0 \mathcal{K}^1 \mathcal{K}^2 \sCurl \vartheta \Phi)_{ABA'B'}\nonumber\\
&  + \tfrac{1}{2} \Psi_{2} \kappa_1^3 (\mathcal{L}_{\xi}G)_{ABA'B'} + (\sTwist \mathcal{A})_{ABA'B'}\label{eq:CurlDgCurlDgphihatEndProof}
\end{align}
by using \eqref{eq:2ndOrderKCommutatorIdentities} (the third term on the right-hand side can be rewritten using \eqref{eq:P12Def}). Since
\begin{align}
\tfrac{1}{2} \nabla_{AA'}\mathcal{A}_{BB'} + \tfrac{1}{2} \nabla_{BB'}\mathcal{A}_{AA'}={}&\tfrac{1}{4} \epsilon_{AB} \bar\epsilon_{A'B'} (\sDiv \mathcal{A}) + (\sTwist \mathcal{A})_{ABA'B'},
\end{align}
\eqref{eq:divA} for the trace terms and \eqref{eq:deltagIrrDec} finally proves the theorem. 
\end{proof} 
We will restrict to the source-free case $\vartheta\Phi_{ABA'B'}=0, \vartheta\Lambda=0$ for the rest of this section. In this case the last term in \eqref{eq:CurlDgCurlDgEK1Spin2ToLie} is zero and $\mathcal{M}_{ABA'B'}$ has the following property.
\begin{corollary} \label{cor:MeinsteinSol}
The complex field $\mathcal{M}_{ABA'B'}$ is a traceless (by definition \eqref{eq:Mabdef:spinor}) solution to the source-free linearized vacuum Einstein equations because the first two terms on the right-hand side of \eqref{eq:CurlDgCurlDgEK1Spin2ToLie} form a linearized diffeomorphism and the third term is a symmetry operator on $\dot{g}_{ABA'B'}$ (remember that $\Psi_2 \kappa_1^3$ is a constant) which is itself a solution.
\end{corollary}
Note that for a complex metric $\MidMet_{ab}$, in addition to trace-free Ricci, $\vartheta\Phi[\MidMet]$, and Ricci scalar curvature, $\vartheta\Lambda[\MidMet]$, we have self dual, $\overline{\vartheta\Psi}[\MidMet]$\footnote{This operator is defined as the complex conjugate of the operator in \eqref{eq:VarSPsiCDe}, i.e.  $\overline{\vartheta\Psi}[G,\slashed{G}]_{A'B'C'D'} = \tfrac{1}{2} (\sCurlDagger \sCurlDagger G)_{A'B'C'D'} - \tfrac{1}{4} \slashed{G}_{} \overline{\Psi}_{A'B'C'D'} $.} and anti-self dual, $\vartheta\Psi[\MidMet]$, curvature and the last two are in general independent (in case of a real metric, they are complex conjugates of each other). We can now compute all curvature components of the complex metric $\MidMet_{ABA'B'}$.
\begin{lemma}\label{lem:Mcurvature} In the source-free case, the curvature, see \eqref{eq:LinGravEqs}, of the complex metric \eqref{eq:Mabdef:spinor} is given by
\begin{subequations}
\begin{align}
\vartheta\Psi[\MidMet]_{ABCD} ={}& \tfrac{1}{2}(\tilde{\mathcal{L}}_{\mathcal{A}}\Psi)_{ABCD} + \tfrac{1}{2}\Psi_{2} \kappa_1^3 (\mathcal{L}_{\xi}\phi)_{ABCD}  \label{eq:MasdWeylCurvature}\\
={}& \tfrac{1}{2}\Psi_{2} \kappa_1^3 (\mathcal{L}_{\xi}\mathcal{P}^{2} \phi)_{ABCD} \label{eq:MasdWeylCurvatureb}\\
\overline{\vartheta\Psi}[\MidMet]_{A'B'C'D'} ={}&  \tfrac{1}{2}(\tilde{\mathcal{L}}_{\mathcal{A}}\bar\Psi)_{A'B'C'D'} +\tfrac{1}{2}\kappa_1^3 \Psi_2 (\mathcal{L}_{\xi}\bar{\phi})_{A'B'C'D'} \label{eq:MsdWeylCurvature}\\
\vartheta\Phi[\MidMet]_{ABA'B'} ={}& 0 \label{eq:MtfRicciCurvature}\\
\vartheta\Lambda[\MidMet] ={}& 0 \label{eq:MRicciScalarCurvature}
\end{align}
\end{subequations}
where $\phi_{ABCD}$ is defined in \eqref{eq:phidef}, and the operator $\tilde{\mathcal{L}}_{\mathcal{A}}$ acting on the self dual Weyl spinor and its complex conjugate is defined by
\begin{subequations}
\begin{align}
(\tilde{\mathcal{L}}_{\mathcal{A}}\bar\Psi)_{A'B'C'D'}={}&\tfrac{1}{2} \bar\Psi_{A'B'C'D'} (\sDiv_{(0,6)} \mathcal{A})
 + 2 \bar\Psi_{(A'B'C'}{}^{F'}(\sCurlDagger_{(0,2)} \mathcal{A})_{D')F'}, \label{eq:LieTildePsiDagger}\\
(\tilde{\mathcal{L}}_{\mathcal{A}}\Psi)_{ABCD}={}&\tfrac{1}{2} \Psi_{ABCD} (\sDiv_{(6,0)} \mathcal{A})
 + 2 \Psi_{(ABC}{}^{F}(\sCurl_{(2,0)} \mathcal{A})_{D)F}.
\end{align}
\end{subequations}
\end{lemma}
\begin{remark} \label{rem:CovTSI}
Expanding the left-hand side of equation \eqref{eq:MsdWeylCurvature}  gives 
\begin{equation}\label{eq:covTSI:typeD}
(\sCurlDagger \sCurlDagger \sCurlDagger \sCurlDagger_{(4,0)} (\kappa_1^4 \mathcal{K}^1 \mathcal{P}^{2} \vartheta\Psi))_{A'B'C'D'} = (\tilde{\mathcal{L}}_{\mathcal{A}}\bar\Psi)_{A'B'C'D'} +\kappa_1^3 \Psi_2 (\mathcal{L}_{\xi}\bar{\phi})_{A'B'C'D'},
\end{equation}
which is the covariant form of the full TSI for source-free linearized gravity on a general Petrov type D vacuum background. See also corollary~\ref{cor:tomain:intro} for a manifestly gauge-invariant form of the covariant TSI on a Kerr background. 
\end{remark}  

\begin{proof}
Commuting two derivatives on $\varphi_{AA'} \in \SymSpinSec_{1,1}$ in vacuum type D spacetime (commutators are given in \cite[\S 2.2]{ABB:symop:2014CQGra..31m5015A}) leads to the operator identity
\begin{align}
(\sCurl \sCurl \sTwist \varphi)_{ABCD}={}&\tfrac{1}{2} \Psi_{ABCD} (\sDiv_{(6,0)} \varphi)
 + 2 \Psi_{(ABC}{}^{F}(\sCurl_{(2,0)} \varphi)_{D)F}.
\end{align}
Using this identity for $\varphi_{AA'} = \mathcal{A}_{AA'}$ and its complex conjugate together with the source-free version of \eqref{eq:CurlDgCurlDgphihatEndProof} and the fact that $\mathcal{M}$ is traceless, the curvature relations \eqref{eq:MasdWeylCurvature} and \eqref{eq:MsdWeylCurvature} follow.

In the source-free case, we find using \eqref{eq:curlK1A}, \eqref{eq:divA} and \eqref{eq:AU} that
\begin{align}
(\tilde{\mathcal{L}}_{\mathcal{A}}\Psi)_{ABCD}={}&\tfrac{3}{16} \Psi_{2} (\mathcal{K}^0 \mathcal{K}^0 \sDiv_{(6,0)} \mathcal{A})_{ABCD}
 + \tfrac{3}{2} \Psi_{2} (\mathcal{K}^0 \mathcal{K}^1 \sCurl_{(2,0)} \mathcal{A})_{ABCD}\nonumber\\
={}&\tfrac{3}{2} \Psi_{2} (\mathcal{K}^0 \sCurl_{(2,0)} \mathcal{K}^1 \mathcal{A})_{ABCD}
 -  \tfrac{3}{8} \Psi_{2} (\mathcal{K}^0 \mathcal{K}^0 \sDiv_{(4,0)} \mathcal{A})_{ABCD} + \tfrac{3}{16} \Psi_{2} (\mathcal{K}^0 \mathcal{K}^0 \sDiv_{(6,0)} \mathcal{A})_{ABCD}\nonumber\\
={}&\tfrac{1}{8} \Psi_{2} \kappa_{1}{}^3 (\mathcal{K}^0 \mathcal{K}^0 \mathcal{L}_{\xi}\mathcal{K}^2 \mathcal{K}^2 \phi)_{ABCD}
 -  \Psi_{2} \kappa_{1}{}^3 (\mathcal{K}^0 \mathcal{L}_{\xi}\mathcal{K}^2 \phi)_{ABCD}\nonumber\\
={}& -  \Psi_{2} \kappa_{1}{}^3 (\mathcal{L}_{\xi}\phi)_{ABCD}
 + \Psi_{2} \kappa_{1}{}^3 (\mathcal{L}_{\xi}\mathcal{P}^{2} \phi)_{ABCD}.
\end{align}
from which \eqref{eq:MasdWeylCurvatureb} follows. 
Equations \eqref{eq:MtfRicciCurvature} and \eqref{eq:MRicciScalarCurvature} follow from corollary~\ref{cor:MeinsteinSol}.
\end{proof}

An important property of the linearized curvature scalars $\vartheta\Psi_0, \vartheta\Psi_4$ entering $\mathcal{M}_{ABA'B'}$ is that they are invariant under infinitesimal diffeomorphisms and so is $\mathcal{M}_{ABA'B'}$ itself. We find the following behavior under gauge transformations.
\begin{lemma} \label{lem:GaugeDependence}
For linearized diffeomorphism of the background metric, generated by a real vector $\nu^{AA'}$ of the original metric we get 
\begin{align}
G_{ABA'B'}={}&2 (\sTwist \nu)_{ABA'B'},&
\slashed{G}_{}={}&2 (\sDiv \nu).
\end{align}
For the curvature and $\mathcal{A}_{AA'}$ we get 
\begin{subequations}
\begin{align}
\vartheta \Phi_{ABA'B'}={}&0,\\
\vartheta \Lambda={}&0,\\
\phi_{ABCD}={}&\tfrac{3}{16} \Psi_{2} (\mathcal{K}^0 \mathcal{K}^0 \sDiv_{(6,0)} \nu)_{ABCD}
 + \tfrac{3}{2} \Psi_{2} (\mathcal{K}^0 \mathcal{K}^1 \sCurl_{(2,0)} \nu)_{ABCD},\\
\mathcal{A}_{AA'}={}&- \Psi_{2} \kappa_1^3 (\mathcal{L}_{\xi}\nu)_{AA'}. \label{eq:Agauge}
\end{align}
\end{subequations}
\end{lemma}
\begin{proof}
For the curvature, one can use the results of \cite{BaeVal15}  and transform it to the operators of this paper using the type D structure of the curvature. Applying $\mathcal{K}^2 \mathcal{K}^2$ or $\mathcal{K}^1 \mathcal{K}^2$ onto $\phi_{ABCD}$ gives
\begin{subequations} 
\begin{align}
(\mathcal{K}^2 \mathcal{K}^2 \phi)_{}={}&\tfrac{1}{2} \Psi_{2} (\sDiv_{(6,0)} \nu)_{},\\
(\mathcal{K}^1 \mathcal{K}^2 \phi)_{AB}={}&\tfrac{3}{2} \Psi_{2} (\mathcal{K}^1 \mathcal{K}^1 \sCurl_{(2,0)} \nu)_{AB} .
\end{align}
\end{subequations} 
These relations can then be used in \eqref{eq:AA11Eq2} to yield
\begin{align}
\mathcal{A}_{AA'}={}&- \Psi_{2} \kappa_1^3 \xi^{BB'} (\mathcal{K}^0 \mathcal{K}^2 \sTwist \nu)_{ABA'B'}
 + \Psi_{2} \kappa_1^3 \xi^{B}{}_{A'} (\mathcal{K}^1 \mathcal{K}^1 \sCurl_{(2,0)} \nu)_{AB}\nonumber\\
& -  \tfrac{1}{2} \Psi_{2} \kappa_1^4 (\mathcal{K}^1 \sTwist_{(2,0)} \sDiv \nu)_{AA'}
 + \tfrac{1}{2} \Psi_{2} \kappa_1^4 (\mathcal{K}^1 \sTwist_{(-1,0)} \sDiv_{(6,0)} \nu)_{AA'}.
\end{align}
An expansion of the extended indices and a reformulation of the Lie derivative in terms of fundamental spinor operators gives the gauge dependence of $\mathcal{A}_{AA'}$. 
\end{proof}

\begin{corollary}\label{cor:tomain:intro}
For linearized diffeomorphism generated by a real vector $\nu^{AA'}$ on a Kerr background, we have
\begin{align}
 \Im \mathcal{A}_{AA'}{}=& 0. \label{eq:ImAgauge}
\end{align}

Then we have the covariant form of the TSI for linearized gravity on Kerr in terms of manifestly gauge invariant quantities,
\begin{align} \label{eq:spin2-TSI}
(\sCurlDagger \sCurlDagger \sCurlDagger \sCurlDagger_{(4,0)} (\kappa_1^4 \mathcal{K}^1 \mathcal{P}^{2} \vartheta\Psi))_{A'B'C'D'}={}&\Psi_{2} \kappa_1^3 (\mathcal{L}_{\xi}\mathcal{P}^{2} \overline{\vartheta\Psi})_{A'B'C'D'}
 + 2i (\tilde{\mathcal{L}}_{\ImA}\overline{\Psi})_{A'B'C'D'},
\end{align} 
where $\tilde{\mathcal{L}}_{\ImA}$ is given by \eqref{eq:LieTildePsiDagger} with $\mathcal{A}$ replaced by $\Im\mathcal{A}$.
\end{corollary}
\begin{proof}
Equation \eqref{eq:ImAgauge} follows from reality of $\xi^{AA'}, \nu^{AA'}$ and $\Psi_2 \kappa_1^3$ on Kerr , see remark~\ref{rem:explicit}.
Expand the curvature operator in \eqref{eq:MsdWeylCurvature} and subtract the complex conjugate of the vacuum identity between \eqref{eq:MasdWeylCurvature} and \eqref{eq:MasdWeylCurvatureb}. To end up with the identity in terms of $\vartheta\Psi_{ABCD}$, use \eqref{eq:phihatDef} and \eqref{eq:phidef}.
\end{proof}

\section{Conclusions and outlook}

The main result, theorem \ref{thm:mainthm}, shows how the Debye potential construction for linearized gravity on a vacuum spacetime of Petrov type D (generalizing the Sachs-Bergmann super-potential for linearized gravity on Minkowski space) can be used to define a complex solution $\MidMet_{ab}$ to the field equations of linearized gravity, which in view of the identity \eqref{eq:CurlDgCurlDgEK1Spin2ToLie}, in the source-free case, is essentially pure gauge. In particular, 
$$
\mathcal{M}_{ABA'B'}- \tfrac{1}{2} \Psi_{2} \kappa_1^3 (\mathcal{L}_{\xi}\dot g)_{ABA'B'}
$$ 
is a pure gauge metric. 

Calculating the self-dual linearized Weyl curvature of $\MidMet_{ab}$ in the source-free case leads to the covariant form \eqref{eq:covTSI:typeD} of the TSI for linearized gravity. It is worth emphasizing that when restricted to the Kerr background, it yields a manifestly gauge invariant form of the TSI given in \eqref{eq:spin2-TSI}, see also appendix~\ref{secApp:GHPForm} for the GHP component form which include three additional equations compared to the classical TSI. In particular, $\Im\mathcal{A}_{AA'}$ is gauge invariant in that case. In \citep{2018PhRvL.121e1104A} two of the authors have presented a new set of gauge invariant quantities for linearized gravity which is complete, cf. \cite{AABKW:complete}. In the proof of completeness, theorem \ref{thm:mainthm} plays a central role. 

Recall that the Bianchi identity in differential geometry is a geometric identity which is valid independently of any field equation. The same is true for the linearized Bianchi identity with sources \eqref{eq:LinBianchi1}, and its specialization to Petrov type D backgrounds \eqref{eq:LinBianchi2}. It is important to note that several of the fundamental identities including \eqref{eq:covTME} and \eqref{eq:CurlDgCurlDgEK1Spin2ToLie} presented here are \emph{operator identities} derived from the Bianchi identity by applying suitable fundamental operators and making use of commutation rules, and are therefore valid for any linearized metric, not necessarily a solution of the linearized Einstein equations. Operator identities as the ones presented in theorem \ref{thm:mainthm} have many interesting applications. For example, they lead to symmetry operators for linearized gravity via the method of adjoint operators, see \cite{2016arXiv160904584A}.

In the paper \citep{2014arXiv1412.2960A} a new conserved super-energy tensor $V_{ab}$ for the Maxwell field on spacetimes of Petrov type D was found. The fact that $V_{ab}$ is conserved can be seen, by direct computation, to be a consequence of the full TSI system for the Maxwell field. In view of the fundamental role played by the full TSI for the spin-1 case in analyzing the conservation laws implied by the Maxwell field equations, we expect that this new covariant identity providing the full TSI for the spin-2 case will lead to a more complete understanding of the structure of linearized gravity on the Kerr spacetime, and in particular the conservation laws implied by this system. 

As mentioned in the introduction, Coll \emph{et al.} in \cite{coll:etal:1987JMP....28.1075C}
discussed an equivalent system in the spin-1 case which consists of the TME and the full TSI in that case. It may be interesting if a similar statement can be made for linearized gravity based on the full TSI system presented here.

\appendix

\section{GHP form of some expressions} \label{secApp:GHPForm}
In this section we present the components of certain spinor equations with respect to a principal null dyad in a Petrov type D spacetime. Recall that in this case, only one of the Newman-Penrose Weyl scalars, $\Psi_2$, is non-zero. We are using the compact GHP formalism \citep{GHP} in which the operators $\tho, \tho', \edt, \edt'$ are weighted directional derivatives along the tetrad. The computations have been performed using the \emph{SpinFrames} package \citep{SpinFrames}.

We note the relation between components of the varied Weyl spinor $\vartheta\Psi_{ABCD}$ used in this paper and the linearized Newman-Penrose Weyl scalars (the difference is due to the variation of the tetrad/dyad).
\begin{align}
\vartheta \Psi_{0} = \dot\Psi_0, &&
\vartheta \Psi_{1} = \dot\Psi_1 + 3 \Psi_{2} (\vartheta o)_{0}, &&
\vartheta \Psi_{2} = \dot\Psi_2,  &&
\vartheta \Psi_{3} = \dot\Psi_3 - 3 \Psi_{2} (\vartheta \iota)_{1}, &&
\vartheta \Psi_{4} = \dot\Psi_4.
\end{align}
The components of the spin-2 TME \eqref{eq:covTME} in the source-free case take the compact form 
\begin{subequations}\label{eq:TME:spin-2}
\begin{align}
 \big((\tho - 3 \rho -\bar{\rho})\tho'   - (\edt - 3 \tau  -  \bar{\tau}')\edt'  - 3 \Psi_{2} \big) (\kappa_1 \vartheta\Psi_0) ={}&0,\\
 \big((\tho' - 3 \rho'-  \bar{\rho}')\tho  - (\edt' -3 \tau'- \bar{\tau}) \edt  - 3 \Psi_{2} \big) (\kappa_1\vartheta\Psi_4)={}&0.
\end{align}
\end{subequations}
The components of the complex vector field $\mathcal{A}_{AA'}$ defined in \eqref{eq:AA11Def} in the source-free case are given by
\begin{subequations}
\begin{align}
\mathcal{A}_{00'}={}&
- \tfrac{3}{4} \slashed{G} \Psi_{2} \kappa_{1}{}^4 \rho
 - 3 G_{11'} \Psi_{2} \kappa_{1}{}^4 \rho
 + 3 G_{10'} \Psi_{2} \kappa_{1}{}^4 \tau
 - 2 \kappa_1^4\vartheta \Psi_{1} \tau'
 + \tho (\kappa_1^4\vartheta \Psi_{2}),\\
\mathcal{A}_{01'}={}&
 -  \tfrac{3}{4} \slashed{G} \Psi_{2} \kappa_{1}{}^4 \tau
-3 G_{12'} \Psi_{2} \kappa_{1}{}^4 \rho
 + 3 G_{11'} \Psi_{2} \kappa_{1}{}^4 \tau
 - 2 \kappa_1^4\vartheta \Psi_{1} \rho'
 + \edt (\kappa_1^4\vartheta \Psi_{2}),\\
\mathcal{A}_{10'}={}&
 \tfrac{3}{4} \slashed{G} \Psi_{2} \kappa_{1}{}^4 \tau'
+  3 G_{10'} \Psi_{2} \kappa_{1}{}^4 \rho'
 - 3 G_{11'} \Psi_{2} \kappa_{1}{}^4 \tau'
+ 2 \kappa_1^4\vartheta \Psi_{3} \rho
 -  \edt' (\kappa_1^4\vartheta \Psi_{2}),\\
\mathcal{A}_{11'}={}&
\tfrac{3}{4} \slashed{G} \Psi_{2} \kappa_{1}{}^4 \rho'
 + 3 G_{11'} \Psi_{2} \kappa_{1}{}^4 \rho'
 - 3 G_{12'} \Psi_{2} \kappa_{1}{}^4 \tau'
 + 2 \kappa_1^4\vartheta \Psi_{3} \tau
 -  \tho' (\kappa_1^4\vartheta \Psi_{2}).
\end{align}
\end{subequations} 
We also state an explicit form of the imaginary part on Kerr, since it is gauge invariant,
\begin{subequations} 
\begin{align}
\ImA_{00'}={}&
- \tfrac{3}{2}i G_{10'} \Psi_{2} \kappa_{1}{}^4 \tau
 -  \tfrac{3}{2}i G_{01'} \Psi_{2} \kappa_{1}{}^4 \tau'
 + i \kappa_1^4\vartheta \Psi_{1} \tau'
 - i \bar{\kappa}_{1'}^4\overline{\vartheta \Psi}_{1'} \bar{\tau}'
 + \tho \Im(\kappa_1^4\vartheta\Psi_2),\\
\ImA_{01'}={}&
\tfrac{3}{2}i G_{12'} \Psi_{2} \kappa_{1}{}^4 \rho
 + \tfrac{3}{2}i G_{01'} \Psi_{2} \kappa_{1}{}^4 \rho'
 + i \kappa_1^4\vartheta \Psi_{1} \rho'
 + i \bar{\kappa}_{1'}^4\overline{\vartheta \Psi}_{3'} \bar{\rho}
 - i \edt \Re(\kappa_1^4\vartheta\Psi_2),\\
\ImA_{10'}={}&
- \tfrac{3}{2}i G_{21'} \Psi_{2} \kappa_{1}{}^4 \rho
 -  \tfrac{3}{2}i G_{10'} \Psi_{2} \kappa_{1}{}^4 \rho'
 - i \kappa_1^4\vartheta \Psi_{3} \rho
 - i \bar{\kappa}_{1'}^4\overline{\vartheta \Psi}_{1'} \bar{\rho}'
 + i \edt' \Re(\kappa_1^4\vartheta\Psi_2),\\
\ImA_{11'}={}&
\tfrac{3}{2}i G_{21'} \Psi_{2} \kappa_{1}{}^4 \tau
 + \tfrac{3}{2}i G_{12'} \Psi_{2} \kappa_{1}{}^4 \tau'
 - i \kappa_1^4\vartheta \Psi_{3} \tau
 + i \bar{\kappa}_{1'}^4\overline{\vartheta \Psi}_{3'} \bar{\tau}
 -  \tho' \Im(\kappa_1^4\vartheta\Psi_2).
\end{align}
\end{subequations}
The dyad components of the full covariant TSI \eqref{eq:spin2-TSI} on a Kerr background are given by 
\begin{subequations}
\begin{align}
0={}&\tho' \tho' \tho' \tho' (\kappa_1^4\vartheta\Psi_{0})
 -  \edt \edt \edt \edt (\kappa_1^4\vartheta\Psi_{4})
 -  \tfrac{M}{27} \mathcal{L}_{\xi}\overline{\vartheta\Psi_4},\\
0={}&\bigl(\edt' (\tho'{}\negthinspace -  \bar{\rho}') - 6 \bar{\rho}' \bar{\tau}\bigr)(\tho'{}\negthinspace + 2 \bar{\rho}')
(\tho'{}\negthinspace + 2 \bar{\rho}')(\kappa_1^4\vartheta\Psi_{0})\nonumber\\
& -  \bigl(\tho (\edt{} -  \bar{\tau}') - 6 \bar{\rho} \bar{\tau}'\bigr)(\edt{} + 2 \bar{\tau}')
 (\edt{} + 2 \bar{\tau}')(\kappa_1^4\vartheta\Psi_{4})\nonumber\\
& - 3i \bar\Psi_{2} (\edt{} -  \tau + 2 \bar{\tau}')\ImA_{11'}
 + 3i \bar\Psi_{2} (\tho'{}\negthinspace -  \rho' + 2 \bar{\rho}')\ImA_{01'},\\
0={}&\bigl((\edt'{}\negthinspace -  \bar{\tau})\tho'{} -12 \bar{\rho}' \bar{\tau}\bigr)
((\edt'{}\negthinspace + 2 \bar{\tau})(\tho'{} + 3 \bar{\rho}') -2 \bar{\rho}' \bar{\tau})(\kappa_1^4\vartheta\Psi_{0})\nonumber\\
& -  \bigl((\edt{} -  \bar{\tau}')\tho{} -12 \bar{\rho} \bar{\tau}'\bigr)
 \bigl((\edt{} + 2 \bar{\tau}')(\tho{} + 3 \bar{\rho})-2 \bar{\rho} \bar{\tau}'\bigr)(\kappa_1^4\vartheta\Psi_{4})\nonumber\\
& + i \bar\Psi_{2} (\edt'{}\negthinspace + 5 \bar{\tau} -  \tau')\ImA_{01'}
 + i \bar\Psi_{2} (\edt{} -  \tau + 5 \bar{\tau}')\ImA_{10'}\nonumber\\
& - i \bar\Psi_{2} (\tho{} -  \rho + 5 \bar{\rho})\ImA_{11'}
  - i \bar\Psi_{2} (\tho'{}\negthinspace -  \rho' + 5 \bar{\rho}')\ImA_{00'},\\
0={}&\bigl(\tho' (\edt'{}\negthinspace -  \bar{\tau}) - 6 \bar{\rho}' \bar{\tau}\bigr)(\edt'{}\negthinspace + 2 \bar{\tau})
(\edt'{}\negthinspace + 2 \bar{\tau})(\kappa_1^4\vartheta\Psi_{0})\nonumber\\
& -  \bigl(\edt (\tho{} -  \bar{\rho}) - 6 \bar{\rho} \bar{\tau}'\bigr)(\tho{} + 2 \bar{\rho})
(\tho{} + 2 \bar{\rho})(\kappa_1^4\vartheta\Psi_{4})\nonumber\\
& - 3i \bar\Psi_{2} (\edt'{}\negthinspace + 2 \bar{\tau} -  \tau')\ImA_{00'}
 + 3i \bar\Psi_{2} (\tho{} -  \rho + 2 \bar{\rho})\ImA_{10'},\\
0={}&\edt' \edt' \edt' \edt' (\kappa_1^4\vartheta\Psi_{0}) -  \tho \tho \tho \tho (\kappa_1^4\vartheta\Psi_{4})
 -  \tfrac{M}{27} \mathcal{L}_{\xi}\overline{\vartheta\Psi_0}.
\end{align}
\end{subequations} 

\begin{remark}
Let $\phi_{AB}$ be a solution to the source-free Maxwell equations and let $\phi_i$, $i=0,1,2$ be the Maxwell scalars. The spin-1 TME is given by
\begin{subequations}\label{eq:TME:spin-1}
\begin{align}
 \big((\tho -\rho - \bar{\rho} )\tho'  - (\edt -  \tau -  \bar{\tau}')\edt'\big) (\kappa_1 \phi_0) ={}&0,\\
 \big((\tho' - \rho' -  \bar{\rho}' )\tho   -  (\edt' -\tau' -\bar{\tau} ) \edt \big) (\kappa_1 \phi_2)={}&0,
\end{align}
\end{subequations}
and the spin-1 extreme TSI are 
\begin{subequations}\label{eq:extremeTSI:spin1}
\begin{align}  
{\edt}'{\edt}' (\kappa_1^2 \phi_0) ={}&   \tho \tho (\kappa_1^2 \phi_2)\\ 
 {\tho}'{\tho}' (\kappa_1^2 \phi_0)={}&  \edt\edt (\kappa_1^2 \phi_2). 
\end{align}
\end{subequations}

For the Maxwell field, the full set of TSI in fact contains a third relation, cf. \cite{coll:etal:1987JMP....28.1075C}, which can be written in the form  
\begin{align}
 (\tho \edt + \bar{\tau}' \tho ) (\kappa_1^2 \phi_2)  ={}& ( \tho' \edt' + \bar{\tau} \tho' ) (\kappa_1^2 \phi_0).
\end{align}
As mentioned above, the full set of TME/TSI equations implied by the Maxwell field equation, has the important consequence that the symmetric tensor $V_{ab}$ introduced in \citep{2014arXiv1412.2960A} is conserved. The tensor $V_{ab}$ is, in contrast to the standard Maxwell stress-energy tensor, independent of the non-radiative modes of the Maxwell field, and is therefore a suitable tool to construct dispersive estimates for the Maxwell field on the Kerr spacetime. 
\end{remark}

\section{$\mathcal{K}$-operator commutators} \label{app:KopComm}

To avoid clutter in the notation we present the commutators as operators which can be applied to arbitrary elements of $\SymSpinSec_(k,l)$, where $k$ and $l$ are large enough so that the combination of operators on the left hand side is properly defined. $v,w$ used below are arbitrary constants. The proof is straightforward but tedious. Examples can be found in Lemma~\ref{lem:Kcommutators}. Complex conjugating these identities gives commutators for the $\overline{\mathcal{K}}$ operators.
\begin{lemma}
Commuting $\mathcal{K}$-operator outside the extended fundamental spinor operators on $\SymSpinSec_{k,l}$ yields
\begin{subequations}
\begin{align}
\sDiv_{(v,w)}\mathcal{K}^2
={}&\mathcal{K}^2 \sDiv_{(v+1,w)},\\
\sCurlDagger_{(v,w)}\mathcal{K}^2
={}&\mathcal{K}^2 \sCurlDagger_{(v+1,w)},\\
\sCurl_{(v,w)}\mathcal{K}^2
={}&\mathcal{K}^2 \sCurl_{(v-1,w)}
-\tfrac{1}{k+1}\mathcal{K}^1 \sDiv_{(v+k,w)},\\
\sTwist_{(v,w)} \mathcal{K}^2
={}&\mathcal{K}^2 \sTwist_{(v-1,w)}
-\tfrac{1}{k+1}\mathcal{K}^1 \sCurlDagger_{(v+k,w)},\\
\sDiv_{(v,w)}\mathcal{K}^1
={}&\tfrac{(k-1)(k+2)}{k(k+1)}\mathcal{K}^1 \sDiv_{(v,w)}
+\tfrac{2}{k}\mathcal{K}^2 \sCurl_{(v-k-1,w)},\\
\sCurlDagger_{(v,w)}\mathcal{K}^1
={}&\tfrac{(k-1)(k+2)}{k(k+1)}\mathcal{K}^1 \sCurlDagger_{(v,w)}
+\tfrac{2}{k}\mathcal{K}^2 \sTwist_{(v-k-1,w)},\\
\sCurl_{(v,w)}\mathcal{K}^1
={}&\mathcal{K}^1 \sCurl_{(v,w)}
+\tfrac{1}{2(k+1)}\mathcal{K}^0 \sDiv_{(v+k+1,w)},\\
\sTwist_{(v,w)}\mathcal{K}^1
={}&\mathcal{K}^1 \sTwist_{(v,w)}
+\tfrac{1}{2(k+1)}\mathcal{K}^0 \sCurlDagger_{(v+k+1,w)},\\
\sDiv_{(v,w)}\mathcal{K}^0
={}&\tfrac{k(k+3)}{(k+1)(k+2)}\mathcal{K}^0 \sDiv_{(v-1,w)}
-\tfrac{4}{k+2}\mathcal{K}^1 \sCurl_{(v-k-2,w)},\\
\sCurlDagger_{(v,w)}\mathcal{K}^0
={}&\tfrac{k(k+3)}{(k+1)(k+2)}\mathcal{K}^0 \sCurlDagger_{(v-1,w)}
-\tfrac{4}{k+2}\mathcal{K}^1 \sTwist_{(v-k-2,w)},\\
\sCurl_{(v,w)}\mathcal{K}^0
={}&\mathcal{K}^0 \sCurl_{(v+1,w)},\\
\sTwist_{(v,w)}\mathcal{K}^0
={}&\mathcal{K}^0 \sTwist_{(v+1,w)}.
\end{align}
\end{subequations}
\end{lemma}
\begin{corollary}
Commuting $\mathcal{K}$-operator inside the extended fundamental spinor operators on $\SymSpinSec_{k,l}$ yields
\begin{subequations}
\begin{align}
\mathcal{K}^2 \sDiv_{(v,w)}
={}&\sDiv_{(v-1,w)}\mathcal{K}^2,\\
\mathcal{K}^2 \sCurlDagger_{(v,w)}
={}&\sCurlDagger_{(v-1,w)}\mathcal{K}^2,\\
\mathcal{K}^2 \sCurl_{(v,w)}
={}&\tfrac{(k-1)(k+2)}{k(k+1)}\sCurl_{(v+1,w)}\mathcal{K}^2
+\tfrac{1}{k+1}\sDiv_{(v+k+1,w)}\mathcal{K}^1,\\
\mathcal{K}^2 \sTwist_{(v,w)}
={}
&\tfrac{(k-1)(k+2)}{k(k+1)}\sTwist_{(v+1,w)}\mathcal{K}^2
+\tfrac{1}{k+1}\sCurlDagger_{(v+k+1,w)}\mathcal{K}^1,\\
\mathcal{K}^1 \sDiv_{(v,w)}
={}
&\sDiv_{(v,w)}\mathcal{K}^1
-\tfrac{2}{k}\sCurl_{(v-k,w)}\mathcal{K}^2,\\
\mathcal{K}^1 \sCurlDagger_{(v,w)}
={}
&\sCurlDagger_{(v,w)}\mathcal{K}^1
-\tfrac{2}{k}\sTwist_{(v-k,w)}\mathcal{K}^2, \\
\mathcal{K}^1 \sCurl_{(v,w)}={}
&\tfrac{k(k+3)}{(k+1)(k+2)}\sCurl_{(v,w)}\mathcal{K}^1
-\tfrac{1}{2(k+1)}\sDiv_{(v+k+2,w)}\mathcal{K}^0,\\
\mathcal{K}^1 \sTwist_{(v,w)}={}
&\tfrac{k(k+3)}{(k+1)(k+2)}\sTwist_{(v,w)}\mathcal{K}^1
-\tfrac{1}{2(k+1)}\sCurlDagger_{(v+k+2,w)}\mathcal{K}^0,\\
\mathcal{K}^0 \sDiv_{(v,w)}={}
&\sDiv_{(v+1,w)}\mathcal{K}^0
+\tfrac{4}{k+2}\sCurl_{(v-k-1,w)}\mathcal{K}^1,\\
\mathcal{K}^0\sCurlDagger_{(v,w)}={}
&\sCurlDagger_{(v+1,w)}\mathcal{K}^0
+\tfrac{4}{k+2}\sTwist_{(v-k-1,w)}\mathcal{K}^1,\\
\mathcal{K}^0\sCurl_{(v,w)}
={}&\sCurl_{(v-1,w)}\mathcal{K}^0,\\
\mathcal{K}^0 \sTwist_{(v,w)}
={}&\sTwist_{(v-1,w)}\mathcal{K}^0.
\end{align}
\end{subequations}
\end{corollary}
\begin{lemma}
For commuting two $\mathcal{K}$-operators on $\SymSpinSec_{k,l}$, we have the identities
\begin{subequations}
\begin{align}
\mathcal{K}^2 \mathcal{K}^0 ={}&
\tfrac{(k-1)(k+2)}{k(k+1)}\mathcal{K}^0 \mathcal{K}^2
-\tfrac{4}{k+2}\mathcal{K}^1 \mathcal{K}^1,\\
\mathcal{K}^1 \mathcal{K}^0 ={}&
\tfrac{k}{k+2}\mathcal{K}^0 \mathcal{K}^1,\\
\mathcal{K}^2 \mathcal{K}^1 ={}&
\tfrac{k-2}{k}\mathcal{K}^1 \mathcal{K}^2,\\
\mathcal{K}^1 \mathcal{K}^1 ={}&
\mathrm{Id} - \tfrac{k-1}{k}\mathcal{K}^0 \mathcal{K}^2,
\end{align}
\end{subequations}
where $\mathrm{Id}$ is the identity operator.
\end{lemma}

\subsection*{Acknowledgements} A substantial part of the work on this paper was carried out during visits by two of the authors, S.A. and T.B. at the Department of Mathematics, Royal Institute of Technology (KTH). We thank the KTH for its support and hospitality. T.B. acknowledges partial support by EPSRC under grant number EP/J011142/1, and L.A. acknowledges support from the Knut and Alice Wallenberg Foundation through a Wallenberg Professorship at KTH. We are grateful to Bernard Whiting for some helpful remarks and for his interest in this work.

\newcommand{\arxivref}[1]{\href{http://www.arxiv.org/abs/#1}{{arXiv.org:#1}}}
\newcommand{\mnras}{Monthly Notices of the Royal Astronomical Society}


\begin{thebibliography}{41}%
\makeatletter
\providecommand \@ifxundefined [1]{%
 \@ifx{#1\undefined}
}%
\providecommand \@ifnum [1]{%
 \ifnum #1\expandafter \@firstoftwo
 \else \expandafter \@secondoftwo
 \fi
}%
\providecommand \@ifx [1]{%
 \ifx #1\expandafter \@firstoftwo
 \else \expandafter \@secondoftwo
 \fi
}%
\providecommand \natexlab [1]{#1}%
\providecommand \enquote  [1]{``#1''}%
\providecommand \bibnamefont  [1]{#1}%
\providecommand \bibfnamefont [1]{#1}%
\providecommand \citenamefont [1]{#1}%
\providecommand \href@noop [0]{\@secondoftwo}%
\providecommand \href [0]{\begingroup \@sanitize@url \@href}%
\providecommand \@href[1]{\@@startlink{#1}\@@href}%
\providecommand \@@href[1]{\endgroup#1\@@endlink}%
\providecommand \@sanitize@url [0]{\catcode `\\12\catcode `\$12\catcode
  `\&12\catcode `\#12\catcode `\^12\catcode `\_12\catcode `\%12\relax}%
\providecommand \@@startlink[1]{}%
\providecommand \@@endlink[0]{}%
\providecommand \url  [0]{\begingroup\@sanitize@url \@url }%
\providecommand \@url [1]{\endgroup\@href {#1}{\urlprefix }}%
\providecommand \urlprefix  [0]{URL }%
\providecommand \Eprint [0]{\href }%
\providecommand \doibase [0]{http://dx.doi.org/}%
\providecommand \selectlanguage [0]{\@gobble}%
\providecommand \bibinfo  [0]{\@secondoftwo}%
\providecommand \bibfield  [0]{\@secondoftwo}%
\providecommand \translation [1]{[#1]}%
\providecommand \BibitemOpen [0]{}%
\providecommand \bibitemStop [0]{}%
\providecommand \bibitemNoStop [0]{.\EOS\space}%
\providecommand \EOS [0]{\spacefactor3000\relax}%
\providecommand \BibitemShut  [1]{\csname bibitem#1\endcsname}%
\let\auto@bib@innerbib\@empty
\bibitem [{\citenamefont {{Kerr}}(1963)}]{kerr:1963PhRvL..11..237K}%
  \BibitemOpen
  \bibfield  {author} {\bibinfo {author} {\bibfnamefont {R.~P.}\ \bibnamefont
  {{Kerr}}},\ }\bibfield  {title} {\enquote {\bibinfo {title} {{Gravitational
  Field of a Spinning Mass as an Example of Algebraically Special Metrics}},}\
  }\href {\doibase 10.1103/PhysRevLett.11.237} {\bibfield  {journal} {\bibinfo
  {journal} {Physical Review Letters}\ }\textbf {\bibinfo {volume} {11}},\
  \bibinfo {pages} {237--238} (\bibinfo {year} {1963})}\BibitemShut {NoStop}%
\bibitem [{\citenamefont {{Andersson}}\ \emph {et~al.}(2015)\citenamefont
  {{Andersson}}, \citenamefont {{B{\"a}ckdahl}},\ and\ \citenamefont
  {{Blue}}}]{2015arXiv150402069A}%
  \BibitemOpen
  \bibfield  {author} {\bibinfo {author} {\bibfnamefont {L.}~\bibnamefont
  {{Andersson}}}, \bibinfo {author} {\bibfnamefont {T.}~\bibnamefont
  {{B{\"a}ckdahl}}}, \ and\ \bibinfo {author} {\bibfnamefont {P.}~\bibnamefont
  {{Blue}}},\ }\bibfield  {title} {\enquote {\bibinfo {title} {{Spin geometry
  and conservation laws in the Kerr spacetime}},}\ }in\ \href {\doibase
  10.4310/SDG.2015.v20.n1.a8} {\emph {\bibinfo {booktitle} {One hundred years
  of general relativity}}},\ \bibinfo {editor} {edited by\ \bibinfo {editor}
  {\bibfnamefont {L.}~\bibnamefont {Bieri}}\ and\ \bibinfo {editor}
  {\bibfnamefont {S.-T.}\ \bibnamefont {Yau}}}\ (\bibinfo  {publisher}
  {International Press},\ \bibinfo {address} {Boston},\ \bibinfo {year}
  {2015})\ pp.\ \bibinfo {pages} {183--226},\ \bibinfo {note}
  {\arxivref{1504.02069}}\BibitemShut {NoStop}%
\bibitem [{\citenamefont {Stephani}\ \emph {et~al.}(2003)\citenamefont
  {Stephani}, \citenamefont {Kramer}, \citenamefont {MacCallum}, \citenamefont
  {Hoenselaers},\ and\ \citenamefont
  {Herlt}}]{stephani:etal:2009esef.book.....S}%
  \BibitemOpen
  \bibfield  {author} {\bibinfo {author} {\bibfnamefont {H.}~\bibnamefont
  {Stephani}}, \bibinfo {author} {\bibfnamefont {D.}~\bibnamefont {Kramer}},
  \bibinfo {author} {\bibfnamefont {M.}~\bibnamefont {MacCallum}}, \bibinfo
  {author} {\bibfnamefont {C.}~\bibnamefont {Hoenselaers}}, \ and\ \bibinfo
  {author} {\bibfnamefont {E.}~\bibnamefont {Herlt}},\ }\href {\doibase
  10.1017/CBO9780511535185} {\emph {\bibinfo {title} {Exact solutions of
  {E}instein's field equations}}},\ \bibinfo {edition} {2nd}\ ed.,\ Cambridge
  Monographs on Mathematical Physics\ (\bibinfo  {publisher} {Cambridge
  University Press, Cambridge},\ \bibinfo {year} {2003})\ pp.\ \bibinfo {pages}
  {xxx+701}\BibitemShut {NoStop}%
\bibitem [{\citenamefont {Penrose}\ and\ \citenamefont
  {Rindler}(1986)}]{Penrose:1986fk}%
  \BibitemOpen
  \bibfield  {author} {\bibinfo {author} {\bibfnamefont {R.}~\bibnamefont
  {Penrose}}\ and\ \bibinfo {author} {\bibfnamefont {W.}~\bibnamefont
  {Rindler}},\ }\href {\doibase 10.1017/CBO9780511564048} {\emph {\bibinfo
  {title} {{Spinors and Space-time I {\&} II}}}},\ Cambridge Monographs on
  Mathematical Physics\ (\bibinfo  {publisher} {Cambridge University Press},\
  \bibinfo {address} {Cambridge},\ \bibinfo {year} {1986})\BibitemShut
  {NoStop}%
\bibitem [{\citenamefont {{Walker}}\ and\ \citenamefont
  {{Penrose}}(1970)}]{walker:penrose:1970CMaPh..18..265W}%
  \BibitemOpen
  \bibfield  {author} {\bibinfo {author} {\bibfnamefont {M.}~\bibnamefont
  {{Walker}}}\ and\ \bibinfo {author} {\bibfnamefont {R.}~\bibnamefont
  {{Penrose}}},\ }\bibfield  {title} {\enquote {\bibinfo {title} {{On quadratic
  first integrals of the geodesic equations for type $\{$2,2$\}$
  spacetimes}},}\ }\href {\doibase 10.1007/BF01649445} {\bibfield  {journal}
  {\bibinfo  {journal} {Communications in Mathematical Physics}\ }\textbf
  {\bibinfo {volume} {18}},\ \bibinfo {pages} {265--274} (\bibinfo {year}
  {1970})}\BibitemShut {NoStop}%
\bibitem [{\citenamefont {{Ferrando}}\ and\ \citenamefont
  {{S{\'a}ez}}(2007)}]{ferrando:saez:2007JMP....48j2504F}%
  \BibitemOpen
  \bibfield  {author} {\bibinfo {author} {\bibfnamefont {J.~J.}\ \bibnamefont
  {{Ferrando}}}\ and\ \bibinfo {author} {\bibfnamefont {J.~A.}\ \bibnamefont
  {{S{\'a}ez}}},\ }\bibfield  {title} {\enquote {\bibinfo {title} {{On the
  invariant symmetries of the $\mathcal{D}$-metrics}},}\ }\href {\doibase
  10.1063/1.2799264} {\bibfield  {journal} {\bibinfo  {journal} {Journal of
  Mathematical Physics}\ }\textbf {\bibinfo {volume} {48}},\ \bibinfo {pages}
  {102504} (\bibinfo {year} {2007})},\ \bibinfo {note}
  {\arxivref{0706.3301}}\BibitemShut {NoStop}%
\bibitem [{\citenamefont {{Carter}}(1968)}]{carter:1968PhRv..174.1559C}%
  \BibitemOpen
  \bibfield  {author} {\bibinfo {author} {\bibfnamefont {B.}~\bibnamefont
  {{Carter}}},\ }\bibfield  {title} {\enquote {\bibinfo {title} {{Global
  Structure of the Kerr Family of Gravitational Fields}},}\ }\href {\doibase
  10.1103/PhysRev.174.1559} {\bibfield  {journal} {\bibinfo  {journal}
  {Physical Review}\ }\textbf {\bibinfo {volume} {174}},\ \bibinfo {pages}
  {1559--1571} (\bibinfo {year} {1968})}\BibitemShut {NoStop}%
\bibitem [{\citenamefont {{Teukolsky}}(1973)}]{teukolsky:1973}%
  \BibitemOpen
  \bibfield  {author} {\bibinfo {author} {\bibfnamefont {S.~A.}\ \bibnamefont
  {{Teukolsky}}},\ }\bibfield  {title} {\enquote {\bibinfo {title}
  {{Perturbations of a Rotating Black Hole. I. Fundamental Equations for
  Gravitational, Electromagnetic, and Neutrino-Field Perturbations}},}\ }\href
  {\doibase 10.1086/152444} {\bibfield  {journal} {\bibinfo  {journal} {\apj}\
  }\textbf {\bibinfo {volume} {185}},\ \bibinfo {pages} {635--648} (\bibinfo
  {year} {1973})}\BibitemShut {NoStop}%
\bibitem [{\citenamefont {{Teukolsky}}\ and\ \citenamefont
  {{Press}}(1974)}]{1974ApJ...193..443T}%
  \BibitemOpen
  \bibfield  {author} {\bibinfo {author} {\bibfnamefont {S.~A.}\ \bibnamefont
  {{Teukolsky}}}\ and\ \bibinfo {author} {\bibfnamefont {W.~H.}\ \bibnamefont
  {{Press}}},\ }\bibfield  {title} {\enquote {\bibinfo {title} {{Perturbations
  of a rotating black hole. III - Interaction of the hole with gravitational
  and electromagnetic radiation}},}\ }\href {\doibase 10.1086/153180}
  {\bibfield  {journal} {\bibinfo  {journal} {\apj}\ }\textbf {\bibinfo
  {volume} {193}},\ \bibinfo {pages} {443--461} (\bibinfo {year}
  {1974})}\BibitemShut {NoStop}%
\bibitem [{\citenamefont {{Starobinski{\v i}}}\ and\ \citenamefont
  {{Churilov}}(1974)}]{1974JETP...38....1S}%
  \BibitemOpen
  \bibfield  {author} {\bibinfo {author} {\bibfnamefont {A.~A.}\ \bibnamefont
  {{Starobinski{\v i}}}}\ and\ \bibinfo {author} {\bibfnamefont {S.~M.}\
  \bibnamefont {{Churilov}}},\ }\bibfield  {title} {\enquote {\bibinfo {title}
  {{Amplification of electromagnetic and gravitational waves scattered by a
  rotating ''black hole''}},}\ }\href@noop {} {\bibfield  {journal} {\bibinfo
  {journal} {Soviet Journal of Experimental and Theoretical Physics}\ }\textbf
  {\bibinfo {volume} {38}},\ \bibinfo {pages} {1} (\bibinfo {year} {1974})},\
  \bibinfo {note}
  {\href{http://www.jetp.ac.ru/cgi-bin/dn/e\_038\_01\_0001.pdf}{http://www.jetp.ac.ru/cgi-bin/dn/e\_038\_01\_0001.pdf}}\BibitemShut
  {NoStop}%
\bibitem [{\citenamefont {{Newman}}\ and\ \citenamefont
  {{Penrose}}(1962)}]{1962JMP.....3..566N}%
  \BibitemOpen
  \bibfield  {author} {\bibinfo {author} {\bibfnamefont {E.}~\bibnamefont
  {{Newman}}}\ and\ \bibinfo {author} {\bibfnamefont {R.}~\bibnamefont
  {{Penrose}}},\ }\bibfield  {title} {\enquote {\bibinfo {title} {An approach
  to gravitational radiation by a method of spin coefficients},}\ }\href
  {\doibase 10.1063/1.1724257} {\bibfield  {journal} {\bibinfo  {journal}
  {Journal of Mathematical Physics}\ }\textbf {\bibinfo {volume} {3}},\
  \bibinfo {pages} {566--578} (\bibinfo {year} {1962})}\BibitemShut {NoStop}%
\bibitem [{\citenamefont {{Geroch}}\ \emph {et~al.}(1973)\citenamefont
  {{Geroch}}, \citenamefont {{Held}},\ and\ \citenamefont {{Penrose}}}]{GHP}%
  \BibitemOpen
  \bibfield  {author} {\bibinfo {author} {\bibfnamefont {R.}~\bibnamefont
  {{Geroch}}}, \bibinfo {author} {\bibfnamefont {A.}~\bibnamefont {{Held}}}, \
  and\ \bibinfo {author} {\bibfnamefont {R.}~\bibnamefont {{Penrose}}},\
  }\bibfield  {title} {\enquote {\bibinfo {title} {{A space-time calculus based
  on pairs of null directions}},}\ }\href {\doibase 10.1063/1.1666410}
  {\bibfield  {journal} {\bibinfo  {journal} {Journal of Mathematical Physics}\
  }\textbf {\bibinfo {volume} {14}},\ \bibinfo {pages} {874--881} (\bibinfo
  {year} {1973})}\BibitemShut {NoStop}%
\bibitem [{\citenamefont {{Andersson}}\ \emph
  {et~al.}(2014{\natexlab{a}})\citenamefont {{Andersson}}, \citenamefont
  {{B{\"a}ckdahl}},\ and\ \citenamefont
  {{Blue}}}]{ABB:symop:2014CQGra..31m5015A}%
  \BibitemOpen
  \bibfield  {author} {\bibinfo {author} {\bibfnamefont {L.}~\bibnamefont
  {{Andersson}}}, \bibinfo {author} {\bibfnamefont {T.}~\bibnamefont
  {{B{\"a}ckdahl}}}, \ and\ \bibinfo {author} {\bibfnamefont {P.}~\bibnamefont
  {{Blue}}},\ }\bibfield  {title} {\enquote {\bibinfo {title} {{Second order
  symmetry operators}},}\ }\href {\doibase 10.1088/0264-9381/31/13/135015}
  {\bibfield  {journal} {\bibinfo  {journal} {Classical and Quantum Gravity}\
  }\textbf {\bibinfo {volume} {31}},\ \bibinfo {eid} {135015} (\bibinfo {year}
  {2014}{\natexlab{a}})},\ \bibinfo {note} {\arxivref{1402.6252}}\BibitemShut
  {NoStop}%
\bibitem [{\citenamefont {{Aksteiner}}\ and\ \citenamefont
  {{B{\"a}ckdahl}}(2016)}]{2016arXiv160904584A}%
  \BibitemOpen
  \bibfield  {author} {\bibinfo {author} {\bibfnamefont {S.}~\bibnamefont
  {{Aksteiner}}}\ and\ \bibinfo {author} {\bibfnamefont {T.}~\bibnamefont
  {{B{\"a}ckdahl}}},\ }\href@noop {} {\enquote {\bibinfo {title} {{Symmetries
  of linearized gravity from adjoint operators}},}\ } (\bibinfo {year}
  {2016}),\ \Eprint {http://arxiv.org/abs/1609.04584} {arXiv:1609.04584
  [gr-qc]} \BibitemShut {NoStop}%
\bibitem [{\citenamefont {{Torres Del
  Castillo}}(1994)}]{TorresDelCastillo:1994}%
  \BibitemOpen
  \bibfield  {author} {\bibinfo {author} {\bibfnamefont {G.~F.}\ \bibnamefont
  {{Torres Del Castillo}}},\ }\bibfield  {title} {\enquote {\bibinfo {title}
  {{Gravitational perturbations of type-D vacuum space-times with cosmological
  constant.}}}\ }\href {\doibase 10.1063/1.530504} {\bibfield  {journal}
  {\bibinfo  {journal} {Journal of Mathematical Physics}\ }\textbf {\bibinfo
  {volume} {35}},\ \bibinfo {pages} {3051--3058} (\bibinfo {year}
  {1994})}\BibitemShut {NoStop}%
\bibitem [{\citenamefont {{Silva-Ortigoza}}(1997)}]{SilvaOrtigoza:1997}%
  \BibitemOpen
  \bibfield  {author} {\bibinfo {author} {\bibfnamefont {G.}~\bibnamefont
  {{Silva-Ortigoza}}},\ }\bibfield  {title} {\enquote {\bibinfo {title} {{A
  Comment on Differential Identities for the Weyl Spinor Perturbations}},}\
  }\href {\doibase 10.1023/A:1018834113248} {\bibfield  {journal} {\bibinfo
  {journal} {General Relativity and Gravitation}\ }\textbf {\bibinfo {volume}
  {29}},\ \bibinfo {pages} {1407--1410} (\bibinfo {year} {1997})}\BibitemShut
  {NoStop}%
\bibitem [{\citenamefont {Whiting}\ and\ \citenamefont
  {Price}(2005)}]{Whiting:2005hr}%
  \BibitemOpen
  \bibfield  {author} {\bibinfo {author} {\bibfnamefont {B.~F.}\ \bibnamefont
  {Whiting}}\ and\ \bibinfo {author} {\bibfnamefont {L.~R.}\ \bibnamefont
  {Price}},\ }\bibfield  {title} {\enquote {\bibinfo {title} {{Metric
  reconstruction from Weyl scalars}},}\ }\href {\doibase
  10.1088/0264-9381/22/15/003} {\bibfield  {journal} {\bibinfo  {journal}
  {Class. Quant. Grav.}\ }\textbf {\bibinfo {volume} {22}},\ \bibinfo {pages}
  {S589--S604} (\bibinfo {year} {2005})}\BibitemShut {NoStop}%
\bibitem [{\citenamefont {{Coll}}\ \emph {et~al.}(1987)\citenamefont {{Coll}},
  \citenamefont {{Fayos}},\ and\ \citenamefont
  {{Ferrando}}}]{coll:etal:1987JMP....28.1075C}%
  \BibitemOpen
  \bibfield  {author} {\bibinfo {author} {\bibfnamefont {B.}~\bibnamefont
  {{Coll}}}, \bibinfo {author} {\bibfnamefont {F.}~\bibnamefont {{Fayos}}}, \
  and\ \bibinfo {author} {\bibfnamefont {J.~J.}\ \bibnamefont {{Ferrando}}},\
  }\bibfield  {title} {\enquote {\bibinfo {title} {{On the electromagnetic
  field and the Teukolsky-Press relations in arbitrary space-times}},}\ }\href
  {\doibase 10.1063/1.527549} {\bibfield  {journal} {\bibinfo  {journal}
  {Journal of Mathematical Physics}\ }\textbf {\bibinfo {volume} {28}},\
  \bibinfo {pages} {1075--1079} (\bibinfo {year} {1987})}\BibitemShut {NoStop}%
\bibitem [{\citenamefont {{Wald}}(1978)}]{wald:1978PhRvL..41..203W}%
  \BibitemOpen
  \bibfield  {author} {\bibinfo {author} {\bibfnamefont {R.~M.}\ \bibnamefont
  {{Wald}}},\ }\bibfield  {title} {\enquote {\bibinfo {title} {{Construction of
  solutions of gravitational, electromagnetic, or other perturbation equations
  from solutions of decoupled equations}},}\ }\href {\doibase
  10.1103/PhysRevLett.41.203} {\bibfield  {journal} {\bibinfo  {journal}
  {Physical Review Letters}\ }\textbf {\bibinfo {volume} {41}},\ \bibinfo
  {pages} {203--206} (\bibinfo {year} {1978})}\BibitemShut {NoStop}%
\bibitem [{Note1()}]{Note1}%
  \BibitemOpen
  \bibinfo {note} {We shall sometimes refer to the linearized vacuum Einstein
  equations as the source-free linearized Einstein equations.}\BibitemShut
  {Stop}%
\bibitem [{\citenamefont {{Sachs}}\ and\ \citenamefont
  {{Bergmann}}(1958)}]{1958PhRv..112..674S}%
  \BibitemOpen
  \bibfield  {author} {\bibinfo {author} {\bibfnamefont {R.}~\bibnamefont
  {{Sachs}}}\ and\ \bibinfo {author} {\bibfnamefont {P.~G.}\ \bibnamefont
  {{Bergmann}}},\ }\bibfield  {title} {\enquote {\bibinfo {title} {{Structure
  of Particles in Linearized Gravitational Theory}},}\ }\href {\doibase
  10.1103/PhysRev.112.674} {\bibfield  {journal} {\bibinfo  {journal} {Physical
  Review}\ }\textbf {\bibinfo {volume} {112}},\ \bibinfo {pages} {674--680}
  (\bibinfo {year} {1958})}\BibitemShut {NoStop}%
\bibitem [{Note2()}]{Note2}%
  \BibitemOpen
  \bibinfo {note} {Here we use a complex anti-self dual Weyl field for
  consistence with the rest of the paper, although this is not used in \cite
  {1958PhRv..112..674S}.}\BibitemShut {Stop}%
\bibitem [{\citenamefont {Penrose}(1965)}]{penrose:1965}%
  \BibitemOpen
  \bibfield  {author} {\bibinfo {author} {\bibfnamefont {R.}~\bibnamefont
  {Penrose}},\ }\bibfield  {title} {\enquote {\bibinfo {title} {{Zero Rest-Mass
  Fields Including Gravitation: Asymptotic Behaviour}},}\ }\href {\doibase
  10.1098/rspa.1965.0058} {\bibfield  {journal} {\bibinfo  {journal} {Royal
  Society of London Proceedings Series A}\ }\textbf {\bibinfo {volume} {284}},\
  \bibinfo {pages} {159--203} (\bibinfo {year} {1965})}\BibitemShut {NoStop}%
\bibitem [{\citenamefont {{Andersson}}\ \emph
  {et~al.}(2014{\natexlab{b}})\citenamefont {{Andersson}}, \citenamefont
  {{B{\"a}ckdahl}},\ and\ \citenamefont {{Joudioux}}}]{2014CMaPh.331..755A}%
  \BibitemOpen
  \bibfield  {author} {\bibinfo {author} {\bibfnamefont {L.}~\bibnamefont
  {{Andersson}}}, \bibinfo {author} {\bibfnamefont {T.}~\bibnamefont
  {{B{\"a}ckdahl}}}, \ and\ \bibinfo {author} {\bibfnamefont {J.}~\bibnamefont
  {{Joudioux}}},\ }\bibfield  {title} {\enquote {\bibinfo {title} {{Hertz
  Potentials and Asymptotic Properties of Massless Fields}},}\ }\href {\doibase
  10.1007/s00220-014-2078-x} {\bibfield  {journal} {\bibinfo  {journal}
  {Communications in Mathematical Physics}\ }\textbf {\bibinfo {volume}
  {331}},\ \bibinfo {pages} {755--803} (\bibinfo {year}
  {2014}{\natexlab{b}})},\ \Eprint {http://arxiv.org/abs/1303.4377}
  {arXiv:1303.4377 [math.AP]} \BibitemShut {NoStop}%
\bibitem [{\citenamefont {{Aksteiner}}\ and\ \citenamefont
  {{Andersson}}(2013)}]{aksteiner:andersson:2013CQGra..30o5016A}%
  \BibitemOpen
  \bibfield  {author} {\bibinfo {author} {\bibfnamefont {S.}~\bibnamefont
  {{Aksteiner}}}\ and\ \bibinfo {author} {\bibfnamefont {L.}~\bibnamefont
  {{Andersson}}},\ }\bibfield  {title} {\enquote {\bibinfo {title} {{Charges
  for linearized gravity}},}\ }\href {\doibase 10.1088/0264-9381/30/15/155016}
  {\bibfield  {journal} {\bibinfo  {journal} {Classical and Quantum Gravity}\
  }\textbf {\bibinfo {volume} {30}},\ \bibinfo {eid} {155016} (\bibinfo {year}
  {2013})},\ \bibinfo {note} {\arxivref{1301.2674}}\BibitemShut {NoStop}%
\bibitem [{Note3()}]{Note3}%
  \BibitemOpen
  \bibinfo {note} {More precisely it differs by a gauge transformation of
  \protect \textit {third kind}, cf. \cite {1979PhRvD..19.1641K}, so that the
  scalar potential solves the TME.}\BibitemShut {Stop}%
\bibitem [{\citenamefont {{Kegeles}}\ and\ \citenamefont
  {{Cohen}}(1979)}]{1979PhRvD..19.1641K}%
  \BibitemOpen
  \bibfield  {author} {\bibinfo {author} {\bibfnamefont {L.~S.}\ \bibnamefont
  {{Kegeles}}}\ and\ \bibinfo {author} {\bibfnamefont {J.~M.}\ \bibnamefont
  {{Cohen}}},\ }\bibfield  {title} {\enquote {\bibinfo {title} {{Constructive
  procedure for perturbations of spacetimes}},}\ }\href {\doibase
  10.1103/PhysRevD.19.1641} {\bibfield  {journal} {\bibinfo  {journal} {\prd}\
  }\textbf {\bibinfo {volume} {19}},\ \bibinfo {pages} {1641--1664} (\bibinfo
  {year} {1979})}\BibitemShut {NoStop}%
\bibitem [{\citenamefont {Aksteiner}(2014)}]{aksteiner:thesis}%
  \BibitemOpen
  \bibfield  {author} {\bibinfo {author} {\bibfnamefont {S.}~\bibnamefont
  {Aksteiner}},\ }\emph {\bibinfo {title} {Geometry and analysis in black hole
  spacetimes}},\ \href@noop {} {Ph.D. thesis},\ \bibinfo  {school} {Gottfried
  Wilhelm Leibniz Universit\"at Hannover} (\bibinfo {year} {2014}),\ \bibinfo
  {note}
  {\href{http://d-nb.info/1057896721}{http://d-nb.info/1057896721}}\BibitemShut
  {NoStop}%
\bibitem [{Note4()}]{Note4}%
  \BibitemOpen
  \bibinfo {note} {The first term on the right-hand side of equation \protect
  \textup {\hbox {\mathsurround \z@ \protect \normalfont (\ignorespaces \ref
  {eq:nice-tensor}\unskip \@@italiccorr )}} is pure gauge since it is the
  action of a linearized diffeomorphism generated by $\protect \EuScript
  {A}_a$.}\BibitemShut {Stop}%
\bibitem [{\citenamefont {{Aksteiner}}\ and\ \citenamefont
  {{B{\"a}ckdahl}}(2018)}]{2018PhRvL.121e1104A}%
  \BibitemOpen
  \bibfield  {author} {\bibinfo {author} {\bibfnamefont {S.}~\bibnamefont
  {{Aksteiner}}}\ and\ \bibinfo {author} {\bibfnamefont {T.}~\bibnamefont
  {{B{\"a}ckdahl}}},\ }\bibfield  {title} {\enquote {\bibinfo {title} {{All
  Local Gauge Invariants for Perturbations of the Kerr Spacetime}},}\ }\href
  {\doibase 10.1103/PhysRevLett.121.051104} {\bibfield  {journal} {\bibinfo
  {journal} {Physical Review Letters}\ }\textbf {\bibinfo {volume} {121}},\
  \bibinfo {eid} {051104} (\bibinfo {year} {2018})},\ \Eprint
  {http://arxiv.org/abs/1803.05341} {arXiv:1803.05341 [gr-qc]} \BibitemShut
  {NoStop}%
\bibitem [{\citenamefont {Geroch}(1970)}]{Ger70spinstructII}%
  \BibitemOpen
  \bibfield  {author} {\bibinfo {author} {\bibfnamefont {R.}~\bibnamefont
  {Geroch}},\ }\bibfield  {title} {\enquote {\bibinfo {title} {Spinor structure
  of space-times in general relativity. {II}},}\ }\href {\doibase
  10.1063/1.1665067} {\bibfield  {journal} {\bibinfo  {journal} {Journal of
  Mathematical Physics}\ }\textbf {\bibinfo {volume} {11}},\ \bibinfo {pages}
  {343--348} (\bibinfo {year} {1970})}\BibitemShut {NoStop}%
\bibitem [{\citenamefont {{B\"{a}ckdahl}}(2011-2016)}]{Bae11a}%
  \BibitemOpen
  \bibfield  {author} {\bibinfo {author} {\bibfnamefont {T.}~\bibnamefont
  {{B\"{a}ckdahl}}},\ }\href@noop {} {\enquote {\bibinfo {title}
  {Sym{M}anipulator},}\ } (\bibinfo {year} {2011-2016}),\ \bibinfo {note}
  {\href{http://www.xact.es/SymManipulator}{http://www.xact.es/SymManipulator}}\BibitemShut
  {NoStop}%
\bibitem [{\citenamefont {Mart\'{\i}n-Garc\'{\i}a}(2002-2016)}]{xAct}%
  \BibitemOpen
  \bibfield  {author} {\bibinfo {author} {\bibfnamefont {J.~M.}\ \bibnamefont
  {Mart\'{\i}n-Garc\'{\i}a}},\ }\href@noop {} {\enquote {\bibinfo {title}
  {x{A}ct: {E}fficient tensor computer algebra for {M}athematica},}\ }
  (\bibinfo {year} {2002-2016}),\ \bibinfo {note}
  {\href{http://www.xact.es}{http://www.xact.es}}\BibitemShut {NoStop}%
\bibitem [{\citenamefont {{Aksteiner}}\ and\ \citenamefont
  {{B\"{a}ckdahl}}(2015)}]{SpinFrames}%
  \BibitemOpen
  \bibfield  {author} {\bibinfo {author} {\bibfnamefont {S.}~\bibnamefont
  {{Aksteiner}}}\ and\ \bibinfo {author} {\bibfnamefont {T.}~\bibnamefont
  {{B\"{a}ckdahl}}},\ }\href@noop {} {\enquote {\bibinfo {title}
  {Spin{F}rames},}\ } (\bibinfo {year} {2015}),\ \bibinfo {note}
  {\href{http://xact.es/SpinFrames/}{http://xact.es/SpinFrames/}}\BibitemShut
  {NoStop}%
\bibitem [{Note5()}]{Note5}%
  \BibitemOpen
  \bibinfo {note} {Note that $\kappa _1$ and $\Psi _2$ can be expressed
  covariantly via the relations $\kappa _{AB} \kappa ^{AB} = -2 \kappa
  _{1}{}^2$ and $\Psi _{ABCD} \Psi ^{ABCD}=6 \Psi _{2}^2$. Hence, we can allow
  $\kappa _1$ and $\Psi _2$ in covariant expressions.}\BibitemShut {Stop}%
\bibitem [{Note6()}]{Note6}%
  \BibitemOpen
  \bibinfo {note} {The name spin raising and lowering is due to the fact that
  multiplication and symmetrization or contraction of a spin-$\protect
  \mathfrak {s}$ field with a valence-2 Killing spinor leads to a
  spin-$\protect \mathfrak {s}+1$ or spin-$\protect \mathfrak {s}-1$ field
  respectively, see \protect \citet [Sec. 6.4]{Penrose:1986fk}.}\BibitemShut
  {Stop}%
\bibitem [{\citenamefont {{B\"ackdahl}}\ and\ \citenamefont {{Valiente
  Kroon}}(2016)}]{BaeVal15}%
  \BibitemOpen
  \bibfield  {author} {\bibinfo {author} {\bibfnamefont {T.}~\bibnamefont
  {{B\"ackdahl}}}\ and\ \bibinfo {author} {\bibfnamefont {J.~A.}\ \bibnamefont
  {{Valiente Kroon}}},\ }\bibfield  {title} {\enquote {\bibinfo {title} {A
  formalism for the calculus of variations with spinors},}\ }\href {\doibase
  10.1063/1.4939562} {\bibfield  {journal} {\bibinfo  {journal} {Journal of
  Mathematical Physics}\ }\textbf {\bibinfo {volume} {57}},\ \bibinfo {eid}
  {022502} (\bibinfo {year} {2016}),\ 10.1063/1.4939562},\ \bibinfo {note}
  {\arxivref{1505.03770}}\BibitemShut {NoStop}%
\bibitem [{Note7()}]{Note7}%
  \BibitemOpen
  \bibinfo {note} {In a type D principal frame this modification only affects
  the middle component.}\BibitemShut {Stop}%
\bibitem [{Note8()}]{Note8}%
  \BibitemOpen
  \bibinfo {note} {This operator is defined as the complex conjugate of the
  operator in \protect \textup {\hbox {\mathsurround \z@ \protect \normalfont
  (\ignorespaces \ref {eq:VarSPsiCDe}\unskip \@@italiccorr )}}, i.e. $\protect
  \overline {\vartheta \Psi }[G,{\mathchoice {{\setbox \z@ \hbox
  {$\mathsurround \z@ \displaystyle G$}\setbox \tw@ \hbox {$\mathsurround \z@
  \displaystyle /$}\dimen 4\wd \z@ \dimen@ \ht \tw@ \advance \dimen@ -\dp \tw@
  \advance \dimen@ -\ht \z@ \advance \dimen@ \dp \z@ \divide \dimen@ \tw@
  \advance \dimen@ -0\ht \tw@ \advance \dimen@ -0\dp \tw@ \dimen@ii 0\wd \z@
  \raise -\dimen@ \hbox to\dimen 4{\hss \kern \dimen@ii \box \tw@ \kern
  -\dimen@ii \hss }\hbox to\z@ {\hss \hbox to\dimen 4{\hss \box \z@ \hss
  }}}}{{\setbox \z@ \hbox {$\mathsurround \z@ \textstyle G$}\setbox \tw@ \hbox
  {$\mathsurround \z@ \textstyle /$}\dimen 4\wd \z@ \dimen@ \ht \tw@ \advance
  \dimen@ -\dp \tw@ \advance \dimen@ -\ht \z@ \advance \dimen@ \dp \z@ \divide
  \dimen@ \tw@ \advance \dimen@ -0\ht \tw@ \advance \dimen@ -0\dp \tw@
  \dimen@ii 0\wd \z@ \raise -\dimen@ \hbox to\dimen 4{\hss \kern \dimen@ii \box
  \tw@ \kern -\dimen@ii \hss }\hbox to\z@ {\hss \hbox to\dimen 4{\hss \box \z@
  \hss }}}}{{\setbox \z@ \hbox {$\mathsurround \z@ \scriptstyle G$}\setbox \tw@
  \hbox {$\mathsurround \z@ \scriptstyle /$}\dimen 4\wd \z@ \dimen@ \ht \tw@
  \advance \dimen@ -\dp \tw@ \advance \dimen@ -\ht \z@ \advance \dimen@ \dp \z@
  \divide \dimen@ \tw@ \advance \dimen@ -0\ht \tw@ \advance \dimen@ -0\dp \tw@
  \dimen@ii 0\wd \z@ \raise -\dimen@ \hbox to\dimen 4{\hss \kern \dimen@ii \box
  \tw@ \kern -\dimen@ii \hss }\hbox to\z@ {\hss \hbox to\dimen 4{\hss \box \z@
  \hss }}}}{{\setbox \z@ \hbox {$\mathsurround \z@ \scriptscriptstyle
  G$}\setbox \tw@ \hbox {$\mathsurround \z@ \scriptscriptstyle /$}\dimen 4\wd
  \z@ \dimen@ \ht \tw@ \advance \dimen@ -\dp \tw@ \advance \dimen@ -\ht \z@
  \advance \dimen@ \dp \z@ \divide \dimen@ \tw@ \advance \dimen@ -0\ht \tw@
  \advance \dimen@ -0\dp \tw@ \dimen@ii 0\wd \z@ \raise -\dimen@ \hbox to\dimen
  4{\hss \kern \dimen@ii \box \tw@ \kern -\dimen@ii \hss }\hbox to\z@ {\hss
  \hbox to\dimen 4{\hss \box \z@ \hss }}}}}]_{A'B'C'D'} = \protect \genfrac
  {}{}{}1{1}{2} (\protect \mathscr {C}^\dagger \protect \mathscr {C}^\dagger
  G)_{A'B'C'D'} - \protect \genfrac {}{}{}1{1}{4} {\mathchoice {{\setbox \z@
  \hbox {$\mathsurround \z@ \displaystyle G$}\setbox \tw@ \hbox {$\mathsurround
  \z@ \displaystyle /$}\dimen 4\wd \z@ \dimen@ \ht \tw@ \advance \dimen@ -\dp
  \tw@ \advance \dimen@ -\ht \z@ \advance \dimen@ \dp \z@ \divide \dimen@ \tw@
  \advance \dimen@ -0\ht \tw@ \advance \dimen@ -0\dp \tw@ \dimen@ii 0\wd \z@
  \raise -\dimen@ \hbox to\dimen 4{\hss \kern \dimen@ii \box \tw@ \kern
  -\dimen@ii \hss }\hbox to\z@ {\hss \hbox to\dimen 4{\hss \box \z@ \hss
  }}}}{{\setbox \z@ \hbox {$\mathsurround \z@ \textstyle G$}\setbox \tw@ \hbox
  {$\mathsurround \z@ \textstyle /$}\dimen 4\wd \z@ \dimen@ \ht \tw@ \advance
  \dimen@ -\dp \tw@ \advance \dimen@ -\ht \z@ \advance \dimen@ \dp \z@ \divide
  \dimen@ \tw@ \advance \dimen@ -0\ht \tw@ \advance \dimen@ -0\dp \tw@
  \dimen@ii 0\wd \z@ \raise -\dimen@ \hbox to\dimen 4{\hss \kern \dimen@ii \box
  \tw@ \kern -\dimen@ii \hss }\hbox to\z@ {\hss \hbox to\dimen 4{\hss \box \z@
  \hss }}}}{{\setbox \z@ \hbox {$\mathsurround \z@ \scriptstyle G$}\setbox \tw@
  \hbox {$\mathsurround \z@ \scriptstyle /$}\dimen 4\wd \z@ \dimen@ \ht \tw@
  \advance \dimen@ -\dp \tw@ \advance \dimen@ -\ht \z@ \advance \dimen@ \dp \z@
  \divide \dimen@ \tw@ \advance \dimen@ -0\ht \tw@ \advance \dimen@ -0\dp \tw@
  \dimen@ii 0\wd \z@ \raise -\dimen@ \hbox to\dimen 4{\hss \kern \dimen@ii \box
  \tw@ \kern -\dimen@ii \hss }\hbox to\z@ {\hss \hbox to\dimen 4{\hss \box \z@
  \hss }}}}{{\setbox \z@ \hbox {$\mathsurround \z@ \scriptscriptstyle
  G$}\setbox \tw@ \hbox {$\mathsurround \z@ \scriptscriptstyle /$}\dimen 4\wd
  \z@ \dimen@ \ht \tw@ \advance \dimen@ -\dp \tw@ \advance \dimen@ -\ht \z@
  \advance \dimen@ \dp \z@ \divide \dimen@ \tw@ \advance \dimen@ -0\ht \tw@
  \advance \dimen@ -0\dp \tw@ \dimen@ii 0\wd \z@ \raise -\dimen@ \hbox to\dimen
  4{\hss \kern \dimen@ii \box \tw@ \kern -\dimen@ii \hss }\hbox to\z@ {\hss
  \hbox to\dimen 4{\hss \box \z@ \hss }}}}}_{} \protect \overline {\Psi
  }_{A'B'C'D'} $.}\BibitemShut {Stop}%
\bibitem [{\citenamefont {{Aksteiner}}\ \emph {et~al.}(2018)\citenamefont
  {{Aksteiner}}, \citenamefont {{Andersson}}, \citenamefont {{B{\"a}ckdahl}},
  \citenamefont {{Khavkine}},\ and\ \citenamefont
  {{Whiting}}}]{AABKW:complete}%
  \BibitemOpen
  \bibfield  {author} {\bibinfo {author} {\bibfnamefont {S.}~\bibnamefont
  {{Aksteiner}}}, \bibinfo {author} {\bibfnamefont {L.}~\bibnamefont
  {{Andersson}}}, \bibinfo {author} {\bibfnamefont {T.}~\bibnamefont
  {{B{\"a}ckdahl}}}, \bibinfo {author} {\bibfnamefont {I.}~\bibnamefont
  {{Khavkine}}}, \ and\ \bibinfo {author} {\bibfnamefont {B.}~\bibnamefont
  {{Whiting}}},\ }\href@noop {} {\enquote {\bibinfo {title} {{Compatibility
  complex for black hole spacetimes}},}\ } (\bibinfo {year} {2018}),\ \bibinfo
  {note} {in preparation}\BibitemShut {NoStop}%
\bibitem [{\citenamefont {{Andersson}}\ \emph {et~al.}(2017)\citenamefont
  {{Andersson}}, \citenamefont {{B{\"a}ckdahl}},\ and\ \citenamefont
  {{Blue}}}]{2014arXiv1412.2960A}%
  \BibitemOpen
  \bibfield  {author} {\bibinfo {author} {\bibfnamefont {L.}~\bibnamefont
  {{Andersson}}}, \bibinfo {author} {\bibfnamefont {T.}~\bibnamefont
  {{B{\"a}ckdahl}}}, \ and\ \bibinfo {author} {\bibfnamefont {P.}~\bibnamefont
  {{Blue}}},\ }\bibfield  {title} {\enquote {\bibinfo {title} {A new tensorial
  conservation law for {Maxwell} fields on the {Kerr} background},}\ }\href
  {\doibase 10.4310/jdg/1486522812} {\bibfield  {journal} {\bibinfo  {journal}
  {J. Differential Geom.}\ }\textbf {\bibinfo {volume} {105}},\ \bibinfo
  {pages} {163--176} (\bibinfo {year} {2017})},\ \bibinfo {note}
  {\arxivref{1412.2960}}\BibitemShut {NoStop}%
\end{thebibliography}

%

\end{document}